\documentclass[10pt,twocolumn]{article}

\usepackage[margin=0.75in,columnsep=0.25in]{geometry}
\usepackage{abstract}
\usepackage[T1]{fontenc}
\usepackage{lmodern}
\usepackage{amsmath,amssymb}
\usepackage{amsthm}
\usepackage{mathtools}
\mathtoolsset{showonlyrefs=true}
\usepackage{bm}
\usepackage{graphicx}
\usepackage{dcolumn}
\usepackage{color}
\providecommand{\colrule}{\hline}
\usepackage[numbers,sort&compress]{natbib}
\usepackage{authblk}

\usepackage[colorlinks=true,linkcolor=blue,citecolor=blue,urlcolor=blue]{hyperref}

\title{{\bf\Large Hidden Phase Transition in Single-Polymer Stretching under Elongational Flow}}

\author[1]{Aiqing Zhu}
\author[2,3]{Fushuai Wang}
\author[4]{Eleonore Vissol-Gaudin}
\author[1,3]{Qianxiao Li\thanks{Contact author: \href{mailto:qianxiao@nus.edu.sg}{qianxiao@nus.edu.sg}}}
\author[2,3]{Beatrice W. Soh\thanks{Contact author: \href{mailto:bwysoh@nus.edu.sg}{bwysoh@nus.edu.sg}}}

\affil[1]{Department of Mathematics, National University of Singapore, 10 Lower Kent Ridge Road, 119076, Singapore}
\affil[2]{Department of Chemical and Biomolecular Engineering, National University of Singapore, 4 Engineering Drive 4, 117585, Singapore}
\affil[3]{Institute for Functional Intelligent Materials, National University of Singapore, 4 Science Drive 2, 117544, Singapore}
\affil[4]{Fujitsu Research of Europe, 3--9 Albert Street, Slough SL1 2BE, United Kingdom}

\date{}

\begin{document}

\twocolumn[
\maketitle
\vspace{-3em}
\begin{onecolabstract}
\vspace{-1em}
Using machine learning, we identify a second transition in polymer chain stretching driven by planar elongational flow, beyond the well-known coil-stretch transition.
We achieve this by developing a conditional extension of Universal and Identifiable Dissipative Dynamics (UIDD) that jointly learns reduced stochastic dynamics from microscopic trajectories across flow strengths and yields the corresponding stationary nonequilibrium free energy landscapes.
As the flow strength increases, the evolving landscape exhibits the expected coil-stretch transition, and then a previously unrecognized dynamical transition. In the latter, the stretched states persist, but the distribution along a folding coordinate becomes bimodal when conditioned on small end-to-end separation.
Large-deviation analysis of a simplified analytical chain model reproduces both transitions and shows that they are phase transitions in the long-chain limit. Single-molecule DNA experiments provide qualitative support for the dynamical transition through the predicted reversal of the folding response.
\end{onecolabstract}
\vspace{1em}
]
\saythanks

{\it {Introduction}---}
The coil-stretch transition in elongational flow is a paradigmatic problem in nonequilibrium polymer dynamics, with important consequences for the rheology of polymer solutions \cite{larson2005rheology, schroeder2018single}. As predicted by de Gennes \cite{de1974coil}, the competition between flow-induced deformation and entropic relaxation drives a sharp change from an entropically coiled state to a highly extended state above a critical flow strength, while conformation-dependent hydrodynamic drag gives rise to hysteresis. Single-molecule experiments have resolved this transition directly \cite{perkins1997single,smith1998response} and observed hysteresis for sufficiently long chains \cite{schroeder2003observation}. Minimal analytical models capture this behavior with a single extension variable  \cite{de1974coil,hinch1977mechanical,larson1989coil}. Although a single chain is a finite system, the transition sharpens with increasing chain length, so coil-stretch behaves as a phase transition in the long-chain limit \cite{de1974coil,schroeder2003observation}.
These models leave open the question of whether coil-stretch fully defines the transition structure of the microscopic chain, or whether additional transitions emerge along conformational degrees of freedom not resolved by extension.

Experiments and simulations suggest that the underlying dynamics are richer than extension alone. Single-molecule experiments reveal ``molecular individualism'' during stretching, in which chains with the same extension adopt qualitatively distinct conformations and evolve along different pathways \cite{perkins1997single,smith1998response}. Brownian dynamics simulations reproduce these conformational classes and show that populations shift with flow strength \cite{larson1999brownian}. This motivates examining how the overall conformational distribution reorganizes with flow strength, without restricting the description to extension alone.

The direct construction of such a description from microscopic dynamics remains difficult, due to the high dimensionality of the underlying polymer system and its many coupled degrees of freedom. Conventional tools each address only part of the problem. Histograms in chosen coordinates presuppose which coordinates matter, while dimensionality reduction methods such as principal component analysis and diffusion maps\cite{coifman2006diffusion} can identify relevant coordinates, but do not yield dynamics or an interpretable potential.
Deep learning offers a way to learn reduced coordinates and their dynamics jointly \cite{chen2018neural,xie2024ab,yu2021onsagernet,zhang2022gfinns}.
Stochastic OnsagerNet demonstrated that effective landscapes can be learned directly from microscopic polymer trajectories~\cite{chen2024constructing}. Universal and Identifiable Dissipative Dynamics (UIDD) further established the identifiability of the learned potential~\cite{zhu2026identifiable}. In previous work, landscapes were learned separately at each flow strength, so the reduced coordinates were not necessarily consistent across conditions, making it difficult to track how the landscape reorganizes and to identify new transitions. Since UIDD identifies the landscape uniquely, extending it to a single model with varying driving conditions would place all landscapes in shared coordinates. Their reorganization with flow strength can then be tracked directly. This turns landscape learning into a systematic tool for detecting transitions, including those invisible to conventional order parameters.

\begin{figure*}
    \centering
     \includegraphics[width=0.95
     \linewidth]{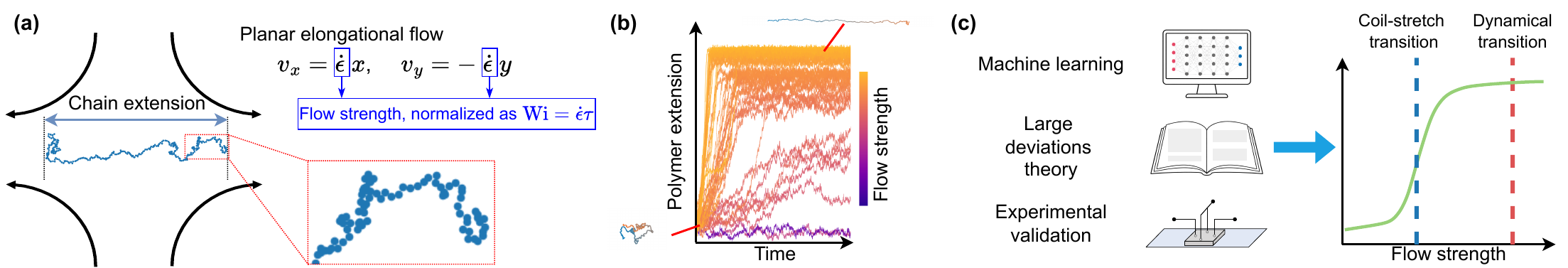} 
    \caption{Overview of the polymer system and transition analysis. ({\bf a}) A single polymer chain is
represented by a bead-rod model, subjected to planar elongational flow, with chain extension measured along the elongational axis. ({\bf b}) Plot of chain
extension as a function of time across varying flow strengths. ({\bf c}) Microscopic trajectories are reduced by PCA and fitted with conditional UIDD to identify phase transitions; large-deviation theory independently supports both transitions, while single-molecule DNA experiments provide qualitative support for the second transition.}
    \label{fig:intro}
\end{figure*}

In this Letter, we develop this approach to revisit the transition structure of a single polymer chain in planar elongational flow (Fig.~\ref{fig:intro}) and show that, beyond coil-stretch, the chain undergoes a second dynamical transition. At higher flow strength, oppositely backfolded saddle states emerge among chains with small end-to-end separation, while the dominant stretched states persist, which is a change invisible in extension.  
We then develop an independent analytical model and derive its nonequilibrium free energy landscapes using large-deviation theory. This simplified model reproduces both transitions and traces them to a competition between conformational entropy and flow-driven axial spreading. In the long-chain limit, the second becomes a phase transition between unfolded and oppositely folded states. Our single-molecule DNA experiments further provide qualitative support for the dynamical transition.

{\it {Kinetic model}---}
The microscopic kinetic model we use is Brownian dynamics for linear, touching-bead polymer chains subjected to an external planar elongational flow of strain rate $\dot\epsilon$ \cite{Soh2018untying}, where $\dot\epsilon$ serves as the control parameter, as illustrated in Fig.~\ref{fig:intro}(a). Each chain consists of $N+1=300$ beads, connected by $N$ rigid rods. The governing stochastic dynamics accounts for excluded volume and Brownian fluctuations, and neglects hydrodynamic interactions (Supplemental Material).
For each microscopic conformation, we record the three-dimensional coordinates of all beads, from which the chain extension $Z_1$ along the elongational axis $x$ is computed. To construct an effective reduced description, we augment $Z_1$ with two additional coordinates \cite{chen2024constructing}, $Z_2$ and $Z_3$, using a symmetry-preserving principal component analysis (PCA).
Correlation analysis associates $Z_2$ and $Z_3$ with the signed end-to-end extension and the signed foldedness, respectively (End Matter).
Each microscopic trajectory at strain rate $\dot\epsilon$ thus yields a time series of $(Z_1, Z_2,Z_3)$.

To learn how dynamics and stationary statistics vary with external conditions, we introduce conditional UIDD, an extension of UIDD~\cite{zhu2026identifiable}. For a $d$-dimensional state $\mathbf Z\in\mathbb R^d$ and condition $c\in\mathcal C$, shared neural functions define a family of It\^{o} SDEs of the form
\begin{equation}\label{eq:uidd-main}
\begin{aligned}
&\mathrm d\mathbf Z_t
=
\Big[\!
-\!\big(M(\mathbf Z_t;c)+W(\mathbf Z_t;c)\big)
\nabla_{\mathbf Z}V(\mathbf Z_t;c)\\
&+\nabla_{\mathbf Z}\!\cdot M(\mathbf Z_t;c)
+\nabla_{\mathbf Z}\!\cdot W(\mathbf Z_t;c)
\Big]\mathrm dt
+
\sigma(\mathbf Z_t;c)\,\mathrm d\mathbf B_t,
\end{aligned}
\end{equation}
where $V$ is a scalar potential
and $M, W, \sigma$ are $d\times d$ matrices depending on state $\mathbf Z$ and condition $c$.
Here, the diffusion $\sigma$ and Brownian motion $\mathbf B_t$ model the thermal fluctuations, and the symmetric matrix $M=\sigma\sigma^\top/2$ characterizes the dissipative response satisfying the fluctuation-dissipation relation.
The conservative parts are encoded by a banded antisymmetric matrix $W$ with nonzero entries only on the first sub- and super-diagonals.
The fields depend on $(\mathbf Z,c)$, while all derivatives act only on $\mathbf Z$. 

This construction possesses desirable structural properties linking the dynamics to its stationary statistics under fixed condition. It is proved that the dynamics retains full expressivity, and the resulting density
\begin{equation}
\rho(\mathbf Z;c)
=\frac{1}{\mathcal Z(c)}\exp\!\left[-V(\mathbf Z;c)\right]
\label{eq:uidd-density}
\end{equation}
satisfies the stationary Fokker-Planck equation, constituting a key advantage over other neural dynamics models \cite{chen2024constructing, desai2021port, lee2021machine, yu2021onsagernet,zhang2022gfinns} (End Matter and Supplemental Material).
Consequently, $V$ is identifiable and coincides, up to an additive constant, with $F=-\ln\rho$. We refer to $V$ (and $F$) as the nonequilibrium free energy; in a driven steady state, its minima and barriers encode the stationary statistics of the chain.
Therefore, changes in the critical point structure of the learned landscape $V(\mathbf Z;c)$ as $c$ varies provide finite-size signatures of transitions.

\begin{figure*}[htbp]
    \centering
    \includegraphics[width=0.95\linewidth]{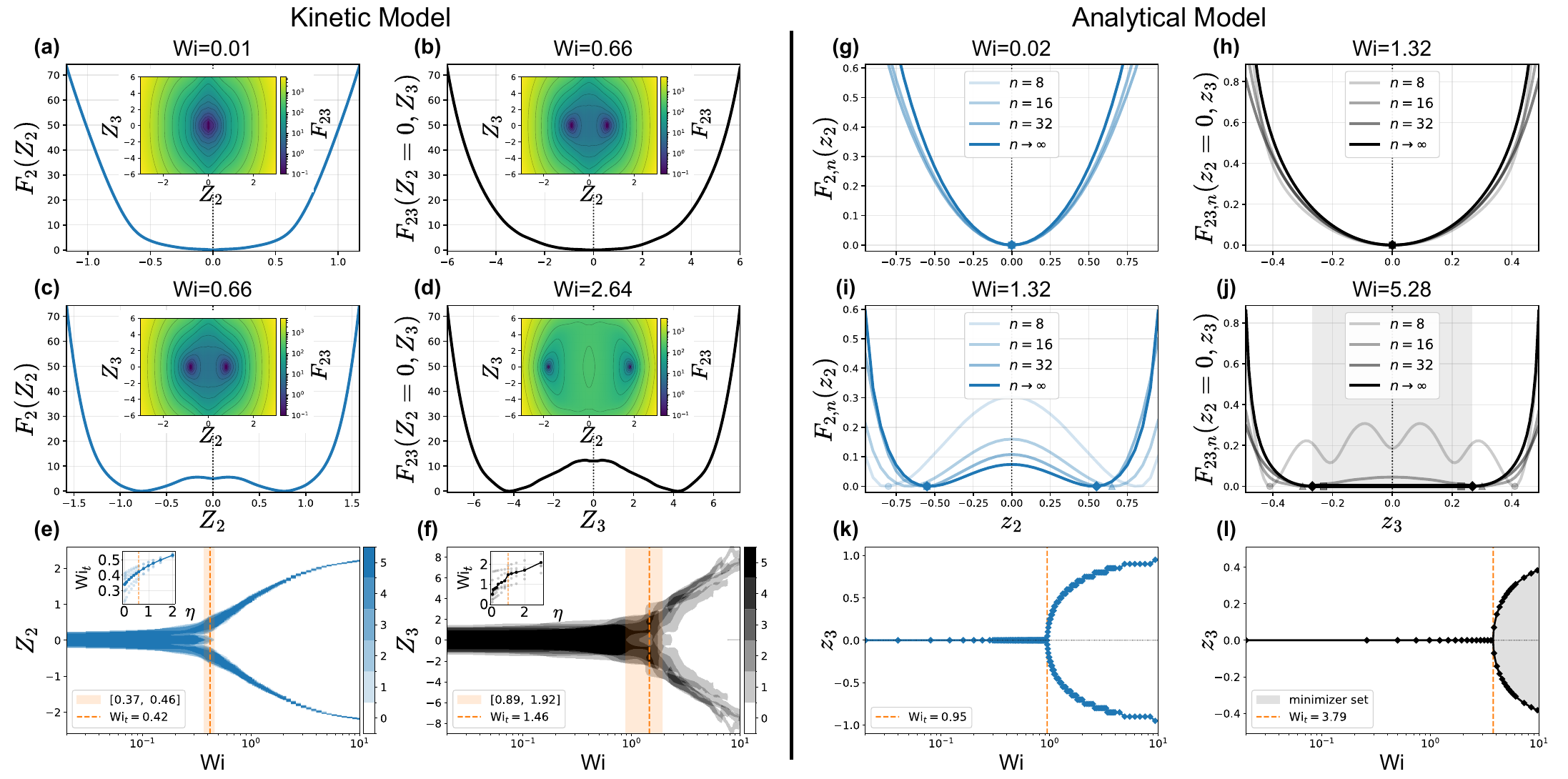}
    
\caption{
Free energy landscapes and transition diagrams.
({\bf a-d}) Learned marginal free energies $F_2(Z_2)$ and sections $F_{23}(0,Z_3)$ 
at the indicated $\mathrm{Wi}$.
Insets show the corresponding 
landscapes $F_{23}(Z_2,Z_3)$.
({\bf e,f}) Near-minimizing sets of the learned $F_2$
($\eta=0.6$) and $F_{23}(0,\cdot)$ ($\eta=1.0$),
respectively, defined by
$\{q:F(q;\mathrm{Wi})-\min_{q'}F(q';\mathrm{Wi})\le\eta\}$.
Color intensity indicates the number of the five
independently trained models satisfying this condition.
For each model, the transition point $\mathrm{Wi}_t$ is determined by the
first crossing of
$F(0;\mathrm{Wi})-\min_q F(q;\mathrm{Wi})$
above $\eta$.
Orange dashed lines and bands indicate the mean and range of $\mathrm{Wi}_t$ across models, respectively, with sensitivities shown in the insets.
({\bf g-j}) Free energy profiles per bond
of the analytical model, $F_{2,n}(z_2)$ 
and $F_{23,n}(0,z_3)$.
({\bf k,l}) Minimizing sets of $F_{2,\infty}(z_2)$
and $F_{23,\infty}(0,z_3)$, respectively.
Beyond the second transition, the gray region in
({\bf l}) denotes the minimizing interval, with diamonds marking its endpoints.
}

    \label{fig:2}
\end{figure*}


Here, we express the flow strength in dimensionless form as the Weissenberg number $\mathrm{Wi}=\dot\varepsilon\tau$, where $\tau$ is the longest relaxation time of the chain. 
We take the condition parameter to be $c= \mathrm{Wi}$ and the reduced state to be $\mathbf Z=(Z_1,Z_2,Z_3)$.
$V, W$ and $\sigma$ are each parameterized as trainable deep neural networks, and their parameters are fitted by jointly maximizing the conditional transition likelihood (End Matter).
The learned dynamics, validated on various initial chain conformations, accurately predict the evolution of chain extension $Z_1$ in all cases (Supplemental Material). 

To examine end-separation statistics and the associated folding structure, 
we introduce the dimensionless marginal free energies
\begin{equation}
\begin{aligned}
F_{23}(Z_2,Z_3;c)
&=-\ln\!\int \exp [-V(\mathbf Z;c)]\,\mathrm dZ_1,\\
F_2(Z_2;c)
&=-\ln\!\iint \exp [-V(\mathbf Z;c)]
\,\mathrm dZ_1\,\mathrm dZ_3,
\end{aligned}
\end{equation}
and the section $F_{23}(0,Z_3;c)$, which probes the landscape along $Z_2=0$ (small-end-separation) and is distinct from a marginal free energy obtained by integrating over $Z_2$.
The results are presented in Fig.~\ref{fig:2} (a-d).

At $\mathrm{Wi}=0.01$, $F_{23}$ has a single minimum at the
origin and $F_2$ exhibits a single well, consistent with a chain in the
coiled state.
By $\mathrm{Wi}=0.66$, the origin becomes a saddle of $F_{23}$, with two dominant minima along $Z_2$; $F_2$ becomes double-welled while
$F_{23}(0,Z_3)$ retains a single minimum at $Z_3=0$. Since $Z_2$ is strongly correlated with the signed end-to-end separation, the two minima correspond to the two orientations of a stretched chain. Hence, the coil-stretch transition appears as a symmetric splitting of the coiled minimum, consistent with the absence of coil-stretch bistability in a model without hydrodynamic interactions \cite{agarwal1998brownian}. 
This reorganization recovers the onset of coil-stretch near the critical value $\mathrm{Wi}_t=0.5$~\cite{de1974coil, perkins1997single} commonly identified through chain extension, here as a change in the nonequilibrium free energy landscape learned from the microscopic kinetic model.

At $\mathrm{Wi}=2.64$, the origin changes from a saddle to a local maximum of $F_{23}$, while two symmetry-related saddles emerge along $Z_3$, appearing as off-center minima of $F_{23}(0,Z_3)$.
This marks a second transition, which splits the small-end-separation population into two oppositely oriented backfolded conformations, while the dominant stretched minima persist. Unlike coil-stretch, this transition changes the unstable states of the landscape instead of the dominant stable states. Since saddles govern the pathways and timescales of unfolding, we refer to it as a dynamical transition, and show below that it becomes a phase transition in the long-chain limit. Previous work~\cite{chen2024constructing} identified prolonged residence near folded state saddles as a key driver of heterogeneous unfolding times \cite{perkins1997single,smith1998response,schroeder2003observation}. 
This transition thus indicates the flow-induced emergence of the unstable states underlying this dynamical heterogeneity, consistent with the growing prevalence of folded unraveling pathways at strong flows. To our knowledge, this transition has not been identified previously, and it is invisible to analyses restricted to stable states

Fig.~\ref{fig:2} (e,f) locate the two transitions through a systematic scan of $\mathrm{Wi}$ in the learned landscapes. Learning errors introduce seed-dependent variations in the free energy profiles, so we track near-minimizing sets, defined as the regions where the free energy lies within a tolerance $\eta$ of its minimum, across five independently trained models. The near-minimizing set first separates into stretched branches along $Z_2$, followed by a separation in the section $F_{23}(0,Z_3)$ at higher flow strength. The stretched branches persist through the dynamical transition, distinguishing the reorganization of the small-end-separation population from the initial coil-stretch transition. At the selected tolerances, the mean transition estimates are $\mathrm{Wi}_t^{(1)}\simeq0.42$ and $\mathrm{Wi}_t^{(2)}\simeq1.46$, with respective ranges $[0.37,0.46]$ and $[0.89,1.92]$. The two ranges do not overlap, confirming that the transitions are distinct.

{\it Analytical model---}
To assess the theoretical plausibility of the two transitions and explore their mechanism, we consider a simplified two-dimensional model. It consists of $n+1$ beads at positions $r_0,\ldots,r_n\in\mathbb R^2$, subject to $\sum_{k=0}^n r_k=0$ and $|r_k-r_{k-1}|=1$ for $k=1,\ldots,n$, with no additional inter-bead interactions.
The dimensionless effective flow potential for a planar elongational flow is
\begin{equation}
U_{n,\lambda}(\{r_k\}_{k=0}^n)=
-\frac{\lambda}{n^2}\sum_{k=0}^n
\left(r_{k,x}^2-r_{k,y}^2\right),
\label{eq:chain-potential}
\end{equation}
where $\lambda\ge0$ controls the elongational strength ($\mathrm{Wi} =0.096\lambda$, see Sec.~\ref{supp:steady-numerics}).
The interpretation of the learned closure coordinates motivates the reduced observables
\begin{equation}
\hat z_2= {(r_{n,x}-r_{0,x})}/{n},
\quad
\hat z_3={(r_{n,x}+r_{0,x})}/{n}.
\label{eq:chain-observables}
\end{equation}
To compare with the learned estimates, we introduce
\begin{equation*}
\begin{aligned}
\mathcal Z_{2,n}(z_2;\lambda)
&=\int e^{-U_{n,\lambda}}
\delta(z_2-\hat z_2)\,\mathrm d\nu_n,\\
\mathcal Z_{23,n}(0,z_3;\lambda)
&=\int e^{-U_{n,\lambda}}
\delta(\hat z_2)\delta(z_3-\hat z_3)\,\mathrm d\nu_n,
\end{aligned}
\label{eq:chain-partition-functions}
\end{equation*}
where $\delta$ denotes the Dirac distribution and $\nu_n$ is the reference configurational measure induced by independent, uniformly distributed bond directions. The corresponding dimensionless free energies per bond are
\begin{equation}
\label{eq: ana finite energy}
F_{2,n}=
- \ln\mathcal Z_{2,n}/n,
\quad
F_{23,n}=
- \ln\mathcal Z_{23,n}/n.
\end{equation}

In the long-chain limit, the scaled interpolated chain belongs to the following subset of the absolutely continuous path space $\mathrm{AC}([0,1];\mathbb R^2)$: 
\begin{equation*}
\mathcal A=
\left\{
m\in\mathrm{AC}([0,1];\mathbb R^2):
m(0)=0,\ |m'(s)|\le1\ \text{a.e.}
\right\}.
\end{equation*}
Writing $\bar m=\int_0^1m(s),\mathrm ds$, we define the effective energy and conformational entropy by
\begin{equation}
\begin{aligned}
\mathcal E[m]
&=-\int_0^1\bigl[(m_x-\bar m_x)^2-(m_y-\bar m_y)^2\bigr] \mathrm ds,\\
\mathcal S[m]
&=-\int_0^1\sup_{h\in\mathbb R} \left\{ h|m'|-\ln I_0(h)\right\} \mathrm ds,
\end{aligned}
\label{eq1}
\end{equation}
where $I_0(h)=(2\pi)^{-1}\int_0^{2\pi}e^{h\cos\theta} \mathrm d\theta$ is the modified Bessel function of the first kind of order zero. By Mogulskii's theorem and Varadhan's lemma \cite{Mogulskii1976,Varadhan2010,AkopyanVysotsky2021}, the limiting free energies are given by (End Matter)
\begin{equation}
\begin{aligned}
F_{2,\infty}(z_2;\lambda)
&=\inf_{\substack{m\in\mathcal A\\ m_x(1)=z_2}}
\left\{\lambda\mathcal E[m]-\mathcal S[m]\right\},\\
F_{23,\infty}(0,z_3;\lambda)
&=\inf_{\substack{m\in\mathcal A,\ m_x(1)=0\\
\bar m_x=-z_3/2}}
\left\{\lambda\mathcal E[m]-\mathcal S[m]\right\}.
\end{aligned}
\label{eq:ana free energy}
\end{equation} 
These variational problems show that the topological changes of free energy minima arise from the competition between the energetic preference for axial spreading and the entropic cost of local alignment.

Fig.~\ref{fig:2} (g-j) presents the variational solutions of \eqref{eq:ana free energy} and Monte Carlo estimates of finite chain free energies \eqref{eq: ana finite energy}. The analytical model qualitatively reproduces both the coil-stretch transition, marked by the splitting of the extension minimum, and the subsequent dynamical transition at zero axial end separation. This sequence arises as increasing $\mathrm{Wi}$ (and hence $\lambda$) strengthens the elongational energy contribution relative to entropy, reshaping the free energy landscape. At $\mathrm{Wi}=5.28$, $F_{23,\infty}(0,z_3)$ is minimized on a symmetric interval containing the origin (Theorem~\ref{supp:minimizing-interval-theorem}), which finite-size effects turn into the two isolated minima seen for $n=16,32$ and in the machine learning results. This flat interval signals phase coexistence between oppositely folded states, making the second transition a phase transition in the long-chain limit. Scans of the variational free energies locate the two transitions at $\mathrm{Wi}_t^{(1)}\simeq0.95$ and $\mathrm{Wi}_t^{(2)}\simeq3.79$ [Fig.~\ref{fig:2} (k,l)], in qualitative agreement with the learned estimates for the full kinetic model.

{\it Experimental signatures---}
To test whether the predicted dynamical transition appears in experiment, we performed single-molecule experiments \cite{soh2023automated} on DNA molecules
stretched in a planar elongational electrokinetic field that is analogous to a planar elongational flow, with the field strength controlled
by the applied voltage (Supplemental Material). The applied voltages of 10-50 V correspond to $\mathrm{Wi}\approx$ [0.7,4.8].
The transition is characterized by changes in the free energy landscape along $Z_3$, but this coordinate is not directly accessible from molecular images.
Motivated by the association of $Z_3$ with the axial displacement of the chain-end midpoint from the center of mass, we define its image-based surrogate $\widetilde Z_3$ as twice the axial displacement of the midpoint between the molecule's axial edges from the intensity-weighted centroid, centered by its trajectory median [Fig.~\ref{fig:exp_normalform}(a)].

\begin{figure}[tbp]
    \centering
    \includegraphics[width=\linewidth]{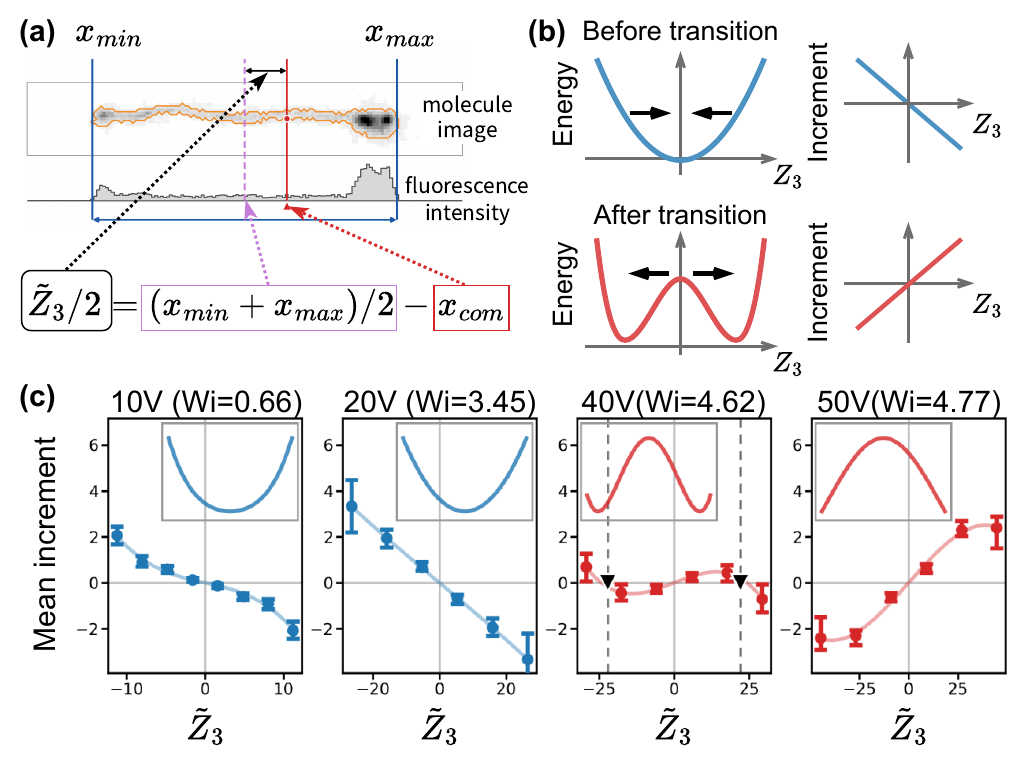}
    \caption{
    ({\bf a}) Definition of the image-based folding coordinate $\widetilde Z_3$.
    ({\bf b}) Schematic of the single-well to double-well transition in the free energy landscape and the corresponding increments near the origin.
    ({\bf c}) Experimental signatures of folding reorganization.
    Symbols show the odd component of the conditional finite-lag mean increment of $\widetilde Z_3$. Error bars indicate uncertainty in the estimated response.
    Curves show polynomial fits; triangles mark zero crossings interpolated from the binned data at 40 V. Insets show apparent potentials obtained as negative
    antiderivatives of the fitted polynomials.}
    \label{fig:exp_normalform}
\end{figure}

Across the second transition, the origin of $F_{23}(0,Z_3)$ changes from a local minimum to a local maximum. Since the drift near the origin follows the local slope of the landscape, small deviations of $Z_3$ from zero should be restored below the transition and amplified above it, reversing the slope of the mean response at the origin from negative to positive [Fig.~\ref{fig:exp_normalform}(b)]. To test this, we measure the mean change in $\widetilde Z_3$ over a fixed time lag, conditioned on its current value, for molecules within a common window of normalized extension during initial stretching, and retain its odd component. As seen in Fig.~\ref{fig:exp_normalform}(c),
at 10 and 20~V, the response is restoring near $\widetilde Z_3=0$, whereas at 40 and 50~V it is repelling, with outer zero crossings at 40~V marking the two folded states, see End Matter for details.
This reversal agrees with the predicted change of the landscape and provides qualitative experimental support for the dynamical transition.

{\it Discussion---}
Beyond the coil-stretch transition, a polymer chain in planar elongational flow undergoes a second dynamical transition, in which backfolded saddle states emerge among chains with small end-to-end separation, while the stretched states persist. To identify this, we developed conditional UIDD to learn stochastic dynamics and free energy landscapes jointly across flow strengths. Large-deviation analysis on a simplified model reproduces both transitions and shows that they are phase transitions in the long-chain limit. Since saddles divide the pathways by which chains stretch, sending a fraction of chains through folded intermediates, this transition offers a landscape-level explanation for why folds become common during unraveling at strong flows \cite{smith1998response,larson1999brownian}. As this transition is invisible in the end-to-end vector, it also identifies a folding coordinate that dumbbell-based constitutive models omit. The landscapes further specify what to measure in experiment, namely a reversal of the folding response from restoring to repelling, which our single-molecule DNA experiments observe.

More broadly, our workflow applies to any system observed through trajectories under a tunable control, such as biomolecules under force or colloidal assembly at varying interaction strength. Since its identifiability does not require detailed balance and it accommodates several conditions at once, it can map transitions across nonequilibrium phase diagrams by tracking how unstable states change across multi-parameter spaces, including those that change unstable states while the most probable states stay the same.

\section*{Methods}
\noindent{\it Closure coordinates---}
In our previous study of the same polymer system  \cite{chen2024constructing,zhu2026identifiable}, two closure coordinates are identified and associated with the axial end-to-end distance and degree of foldedness, respectively. These coordinates were obtained from a nonlinear PCA-ResNet encoder and interpreted a posteriori through geometric correlations \cite{chen2024constructing}.
To compare conformational statistics and free energy landscapes across flow strength without ambiguity from changes in the representation itself, we construct a common symmetry-preserving PCA representation with prescribed parity under the physical symmetries of the chain.
Here $\mathbf r_i=(r_{i,x},r_{i,y},r_{i,z})$ is the position of bead $i=0,\ldots,N$, with $0$ and $N$ labeling the chain ends. The $x$ axis is the elongational direction, and $r_{\mathrm{com},x}=(N+1)^{-1}\sum_{i=0}^{N}r_{i,x}$ is the axial center-of-mass coordinate.
Specifically, we impose chain-end exchange $R:\mathbf r_i\mapsto\mathbf r_{N-i}$ and the three Cartesian coordinate reflections $T_\alpha:r_{i,\alpha}\mapsto-r_{i,\alpha}$, $\alpha=x,y,z$.
The coordinate $Z_2$ is odd under $R$ and $T_x$ and even under $T_y$ and $T_z$, whereas $Z_3$ is even under $R$, odd under $T_x$, and even under $T_y$ and $T_z$. A single projection and normalization are used over all data. Details are given in Supplemental Material. 

\begin{figure}[htbp]
\centering
\includegraphics[width=\linewidth]{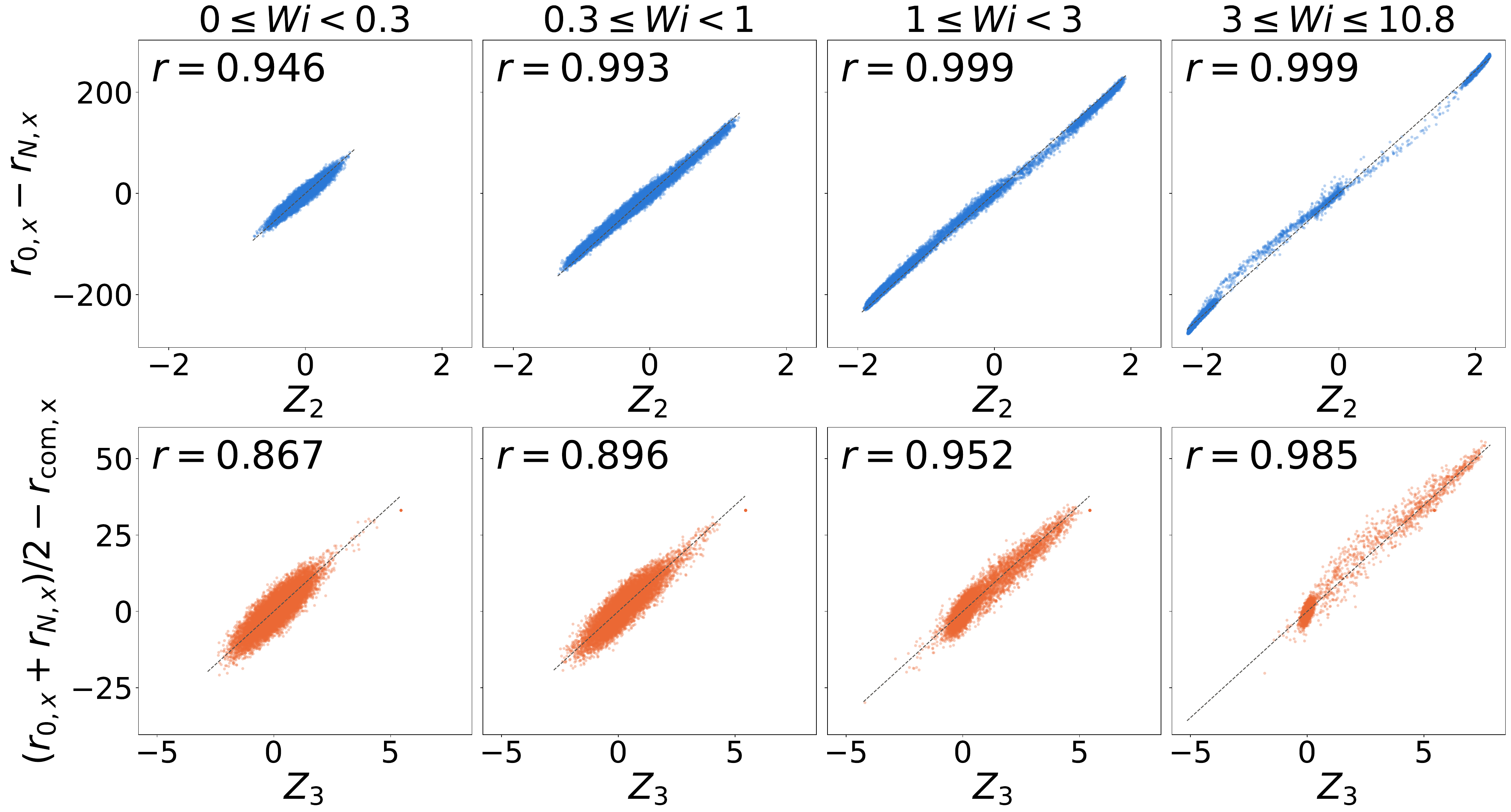}
\caption{
Correlations of the PCA coordinates.
Top: $Z_2$ versus the signed axial endpoint difference,
$r_{0,x}-r_{N,x}$.
Bottom: $Z_3$ versus the signed foldedness,
$(r_{0,x}+r_{N,x})/2-r_{\mathrm{com},x}$.
Correlation coefficients are computed separately for each $\mathrm{Wi}$ interval.
}
\label{fig:em-coordinate-correlations}
\end{figure}

We then test whether these symmetry-adapted coordinates retain the geometric interpretations across flow strengths by evaluating the corresponding correlations separately. As shown in Fig.~\ref{fig:em-coordinate-correlations}, $Z_2$ is strongly correlated with the signed axial endpoint difference, whereas $Z_3$ is strongly correlated with the signed foldedness.

\noindent{\it Conditional UIDD---}
Conditional UIDD learns a family of stochastic dynamics on
$\mathbb R^d$ from trajectories labeled by a fixed external condition $c\in\mathcal C\subset
\mathbb R^p$. It couples the inference problems at different conditions through shared state-and-condition maps $V$, $W$, and $\sigma$. For each fixed $c$, normalizability and
dissipation conditions \cite{huang2015steady} yield the existence and uniqueness of the stationary density for
Eq.~\eqref{eq:uidd-main}. Substituting $\rho$ from Eq.~\eqref{eq:uidd-density}
into the Fokker-Planck current \cite{pavliotis2014stochastic} we have 
\begin{equation}
\begin{aligned}
\mathbf J
=f\rho
-\nabla_{\mathbf Z}\!\cdot(M\rho) =\nabla_{\mathbf Z}\!\cdot(W\rho),
\end{aligned}
\label{eq:em-uidd-current}
\end{equation}
where $(\nabla_{\mathbf Z}\cdot A)_i=\sum_j\partial_{Z_j}A_{ij}$.
The tensor $W\rho$ is antisymmetric,
so its double divergence vanishes and thus
\begin{equation} 
\nabla_{\mathbf Z}\!\cdot\mathbf J=0.
\end{equation}
Consequently, the energy potential within each condition are identifiable, although the scalar reference $C(c)$ need not be unique. 

We first split the time series into multiple transitions.
For each transition $a$, the state $\mathbf Z_a$, condition $c_a$, increment $\Delta\mathbf Z_a$, and observation interval $\Delta t_a$ enter a joint conditional likelihood. With $f$ denoting the drift of Eq.~\eqref{eq:uidd-main}, define $e_a=\Delta\mathbf Z_a-f(\mathbf Z_a;c_a)\Delta t_a$ and the positive-definite covariance $C_a=2M(\mathbf Z_a;c_a)\Delta t_a$.
The Euler-Maruyama approximation gives, up to parameter-independent terms, the objective \cite{pavliotis2014stochastic,dietrich2023learning,zhu2024dyngma}
\begin{equation}
\mathcal L(\theta)
=\frac12\sum_a
\left[
e_a^\top C_a^{-1}e_a+\ln\det C_a
\right].
\label{eq:em-uidd-loss}
\end{equation}
Minimizing this objective over all condition-labeled transitions fits the parameters $\theta$. Training requires neither stationary initial samples nor evaluation of the normalization $\mathcal Z(c)$.
The proof of universality of the conditional UIDD representation, the network construction, training settings, and landscape evaluation are given in Supplemental Material. 

\noindent{\it Large-deviation derivation---}
Fix $\lambda\ge0$ and linearly interpolate
$m_n(k/n)=(r_k-r_0)/n$. These paths lie in $\mathcal A$;
write $\mathbb P_n$ for their distribution under $\nu_n$.
The unit bond directions are independent and uniform, with log
moment-generating function
\begin{equation*}
\Lambda(\mathbf h)=\ln\langle
e^{\mathbf h\cdot(r_1-r_0)}\rangle_{\nu_n}
=\ln I_0(|\mathbf h|).
\end{equation*}
Its Legendre transform is therefore radial:
\begin{equation}
 \Lambda^*(v)=\ell(|v|),\qquad
 \ell(u)=\sup_{h\in\mathbb R}\{hu-\ln I_0(h)\}.
 \label{eq:em-bond-entropy}
\end{equation}
Because $\Lambda$ is finite on $\mathbb R^2$,
Mogulskii's theorem~\cite{Mogulskii1976} gives a path large-deviation principle (LDP) in the uniform topology, with
speed $n$ and conformational entropy~\eqref{eq1}: 
\begin{equation}
 \mathcal S[m]=-\int_0^1\ell(|m'(s)|)\,\mathrm ds,
 \qquad m\in\mathcal A,
 \label{eq:em-path-entropy}
\end{equation}
with $\mathcal S[m]=-\infty$ outside $\mathcal A$.

The zero center-of-mass condition gives the exact relation
\[
 \frac{r_k}{n}=m_n(k/n)-\frac1{n+1}\sum_{j=0}^n m_n(j/n).
\]
Substitution into Eq.~\eqref{eq:chain-potential} gives
\begin{equation}
\begin{aligned}
 -U_{n,\lambda}&=-\lambda n\mathcal E[m_n]+o(n),
 \qquad \hat z_2=m_{n,x}(1),\\
 \mathcal E[m]&=-\int_0^1
 \bigl[(m_x-\bar m_x)^2-(m_y-\bar m_y)^2\bigr]\,\mathrm ds.
\end{aligned}
\label{eq:em-flow-energy}
\end{equation}
The error is uniform because the paths are $1$-Lipschitz.
Thus the $n^{-2}$ prefactor makes the flow exponent scale with the LDP speed $n$. For $|z_2|<1$, integrating the partition density over a window of fixed width $2\varepsilon$ gives
\begin{equation}
\begin{aligned}
 W_{2,n,\varepsilon}(z_2;\lambda)
 &=\int e^{-U_{n,\lambda}}
 \mathbf1_{\{|\hat z_2-z_2|<\varepsilon\}}\,\mathrm d\nu_n\\
 &\asymp\int_{\substack{m\in\mathcal A\\|m_x(1)-z_2|<\varepsilon}}
 e^{-n\lambda\mathcal E[m]}\,\mathrm d\mathbb P_n(m),
\end{aligned}
\label{eq:em-constraint-window}
\end{equation}
where $\asymp$ means that the logarithms differ by $o(n)$.
The energy $\mathcal E$ is bounded and continuous on $\mathcal A$.
The constrained Laplace bounds~\cite{Varadhan2010} therefore give
\begin{equation}
 -\lim_{n\to\infty}\frac1n\ln W_{2,n,\varepsilon}
 =\inf_{\substack{m\in\mathcal A\\|m_x(1)-z_2|<\varepsilon}}
 \{\lambda\mathcal E[m]-\mathcal S[m]\}.
 \label{eq:em-constrained-partition}
\end{equation}
At each sufficiently large fixed $n$, the continuous density satisfies
$W_{2,n,\varepsilon}/(2\varepsilon)\to\mathcal Z_{2,n}(z_2;\lambda)$
as $\varepsilon\downarrow0$. Theorem~\ref{supp:limit-interchange}
in the Supplemental Material proves that the $\varepsilon$ and $n$
limits commute at the logarithmic scale. Consequently,
\begin{equation}
\begin{aligned}
 F_{2,\infty}(z_2;\lambda)
 &=-\lim_{\varepsilon\downarrow0}\lim_{n\to\infty}
 \frac1n\ln\!\left[
 \frac{W_{2,n,\varepsilon}(z_2;\lambda)}{2\varepsilon}
 \right]\\
 &=\inf_{\substack{m\in\mathcal A\\m_x(1)=z_2}}
 \{\lambda\mathcal E[m]-\mathcal S[m]\}.
\end{aligned}
\label{eq:em-f2-limit}
\end{equation}
The section $F_{23,\infty}(0,z_3;\lambda)$ is obtained in the same way.

\noindent{\it Conditional response of image asymmetry---}
We examine how the mean response of image asymmetry changes with
increasing voltage (Fig.~\ref{fig:exp_normalform}).
To compare similar stages of initial stretching, reference frames are
selected within a common window of normalized image span. The raw dataset comprises $48$ trajectories ($47{,}339$ frames) at $10$~V, $201$ ($34{,}716$ frames) at $20$~V, $335$ ($27{,}922$ frames) at $40$~V, and $448$ ($33{,}726$ frames) at $50$~V.

For each image, let $\mathcal M$ be the retained molecular mask,
$x_{\min},x_{\max}$ its axial extrema, and $I_p$ the background-corrected
intensity at pixel coordinate $x_p$. Define
\begin{equation}
\begin{aligned}
&x_{\mathrm{com}} =\frac{\sum_{p\in\mathcal M}I_px_p}{\sum_{p\in\mathcal M}I_p},\quad
\frac q2 =\frac{x_{\min}+x_{\max}}2-x_{\mathrm{com}},\\
&\widetilde Z_3^{(i)}(t)
=q_i(t)-\mathop{\mathrm{median}}_{s\in\mathcal Q_i}q_i(s).
\end{aligned}
\label{eq:em-exp-coordinate}
\end{equation}
Here $\mathcal Q_i$ contains all quality-filtered frames of trajectory $i$
(Sec.~\ref{supp:exp-platform}); coordinates are in pixels.
Such image-based coordinate \(\tilde Z_{3}\) correlates with the microscopic folding asymmetry $(r_{0,x}+r_{N,x})/2-r_{\mathrm{com},x}$, as shown in Fig.~\ref{fig:image-based coordinates}.
For conditioning, set $z(t)=\widetilde Z_3(t)-\mu$, where $\mu$ is the
mean reference value at each voltage. This shift leaves increments unchanged. 
For two later frames, $0<\delta_1<\delta_2$, we define
\begin{equation}
\begin{aligned}
\Delta\widetilde Z_3(t)
=&\widetilde Z_3(t+\delta_2) -\widetilde Z_3(t+\delta_1),\\
f(z) =&\left\langle\Delta\widetilde Z_3(t)\!\mid z(t)=z\right\rangle,\\ 
g(z) =&\frac{f(z)-f(-z)}{2}.
\end{aligned}
\label{eq:em-exp-response}
\end{equation}
The use of two later frames separates the conditioning measurement
from the finite-lag increment. The antisymmetrization extracts the component odd under $z\mapsto -z$, thereby focusing the analysis on the restoring or amplifying tendency shared by positive and negative deviations from the chosen center.

\begin{figure}[!htbp]
    \centering
    \includegraphics[width=\linewidth]{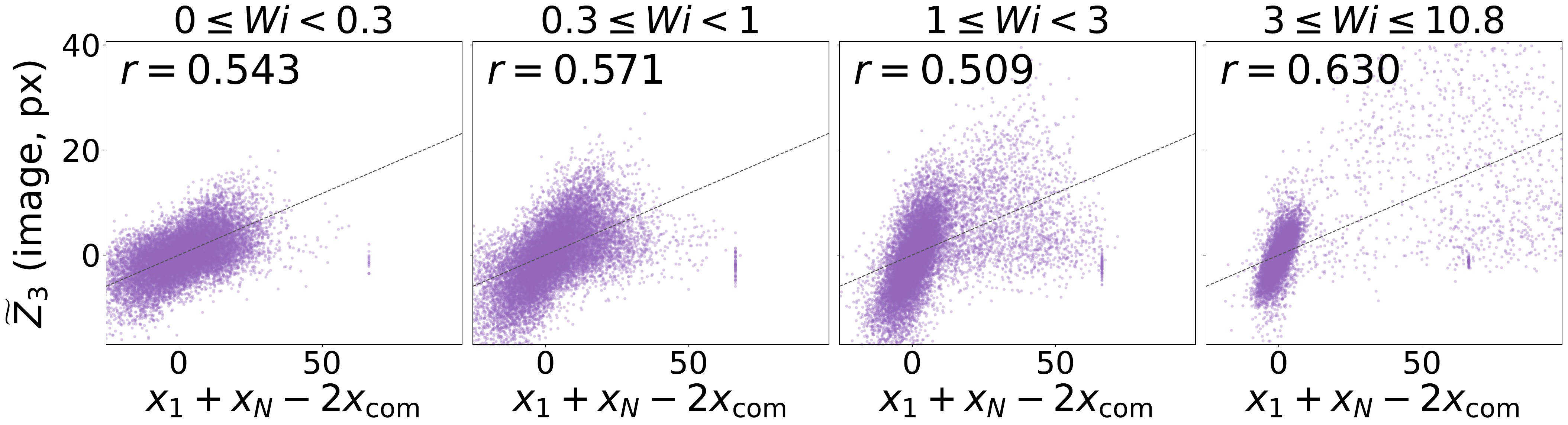}
    \caption{Correlations of the image-based coordinates. The image-based coordinate $\tilde Z_{3}$ is extracted from images generated from simulated polymer conformations, whereas the microscopic folding asymmetry $x_1+x_N-2x_{\mathrm{com}}$ are computed directly from the simulated configurations.}
    \label{fig:image-based coordinates}
\end{figure}

We estimate $g$ in symmetric bins and quantify uncertainty by
bootstrap resampling entire trajectories within each voltage ensemble,
with the bin edges held fixed. The error bars in
Fig.~\ref{fig:exp_normalform} denote pointwise $16$th-$84$th
percentile intervals from $2000$ bootstrap resamples.
The resampling specification and its remaining implementation details are discussed in
Sec.~\ref{supp:exp-response}.
The marked zero crossings at $40$~V are obtained by linear interpolation
between adjacent odd-response bins.

The smooth curves in Fig.~\ref{fig:exp_normalform} are unweighted
least-squares fits of
\[
g_{\rm fit}(z)=c_1z+c_3z^3
\]
to the retained bins, with no sign constraints on the coefficients.
The corresponding apparent potentials are defined by
\begin{equation}
U_{\rm app}(z)
=-\int_0^z g_{\rm fit}(u)\,\mathrm du+C
=-\frac{c_1z^2}{2}-\frac{c_3z^4}{4}+C.
\label{eq:em-exp-potential}
\end{equation}
Each inset is shifted to zero minimum over the displayed range and
rescaled independently, so the comparison concerns the evolution of
potential shape across voltages. These apparent potentials provide a
compact representation of the fitted finite-lag responses.
Experimental details, selection and lag parameters, and sensitivity
analyses are given in Supplemental Material.

\bibliographystyle{unsrtnat}
\bibliography{main}

\clearpage
\onecolumn
\setcounter{section}{0}
\setcounter{subsection}{0}
\setcounter{equation}{0}
\setcounter{figure}{0}
\setcounter{table}{0}
\setcounter{secnumdepth}{2}
\renewcommand{\thesection}{S\arabic{section}}
\renewcommand{\thesubsection}{\thesection.\Alph{subsection}}
\makeatletter
\renewcommand{\p@subsection}{}
\makeatother
\renewcommand{\theequation}{S\arabic{equation}}
\renewcommand{\thefigure}{S\arabic{figure}}
\renewcommand{\thetable}{S\arabic{table}}
\begin{center}
\Large\textbf{Supplemental Material for\\
Hidden Phase Transition in Single-Polymer Stretching under Elongational Flow}
\end{center}

\section{Microscopic dynamics and reduced coordinates}
\label{supp:kinetic-model}

\subsection{Bead-rod chain and its numerical integration}
\label{supp:kinetic-dynamics}
We generate the microscopic trajectories using a Brownian dynamics model of a linear chain of $N+1=300$ touching beads in three dimensions \cite{Soh2018untying}. Each bead has diameter $b=10\mathrm{nm}$, and adjacent beads are connected by rigid rods of the same length:
\begin{equation}\label{eq:supp-kinetic-rods}
|\mathbf r_{i+1}-\mathbf r_i|=b, \qquad i=0,\ldots,N-1.
\end{equation}
where the beads are indexed by $i=0,\ldots,N$, $\mathbf r_i(t)$ is the position of bead $i$ and $\mathbf r_{\mathrm{cm}}(t)=(N+1)^{-1}\sum_{i=0}^{N}\mathbf r_i(t)$ is the instantaneous center of mass. The imposed planar elongational flow is evaluated relative to the center of mass:
\begin{equation}
 \mathbf u_i(t)
 =\boldsymbol\kappa
  [\mathbf r_i(t)-\mathbf r_{\mathrm{cm}}(t)],
 \qquad
 \boldsymbol\kappa
 =\dot\epsilon
 \begin{pmatrix}1&0&0\\0&-1&0\\0&0&0\end{pmatrix}.
 \label{eq:supp-kinetic-flow}
\end{equation}
The stagnation point follows the instantaneous center of mass, with extension along the $x$-axis, compression along the $y$-axis and no imposed velocity in the $z$-axis. 

Let $\zeta=3\pi\eta b$ denote the Stokes drag coefficient of a single bead in a solvent of viscosity $\eta$. We neglect hydrodynamic interactions between beads, such that the solvent acts only through the imposed flow and the bead-wise drag coefficient. 
Neglecting inertia, the force balance is
\begin{equation}\label{eq:supp-kinetic-force-balance}
\zeta\dot{\mathbf r}_i=\zeta\mathbf u_i+\mathbf F_i^{\mathrm{ev}}+\mathbf F_i^{\mathrm c}+\mathbf F_i^{\mathrm B},
\end{equation}
where $\mathbf F_i^{\mathrm{ev}}$ is the excluded volume force, $\mathbf F_i^{\mathrm c}$ enforces Eq.~\eqref{eq:supp-kinetic-rods} and $\mathbf F_i^{\mathrm B}$ is the random Brownian force. With $r_{ij}=|\mathbf r_i-\mathbf r_j|$, the short-range contact repulsion is given by
\begin{equation}\label{eq:supp-kinetic-ev}
\mathbf F_i^{\mathrm{ev}} =\mu\sum_{\substack{j\ne i,\ r_{ij}<b}} [\mathbf r_i-\mathbf r_j]/{r_{ij}},
\end{equation}
where $\mu = 35\ \mathrm{pN}$ is chosen to suppress chain crossings \cite{vologodskii2006brownian}.
For \(r_{ij}\neq b\), this force follows from the continuous pair potential
$E_{\mathrm{ev}}=\mu\sum_{i<j}(b-r_{ij})_+$, where $(a)_+=\max(a,0)$.
Defining the bond unit vector $\widehat{\mathbf b}_i=(\mathbf r_{i+1}-\mathbf r_i)/b$ and setting $T_{-1}=T_N=0$ (with the boundary terms omitted), the constraint force on bead $i$ is
\begin{equation}\label{eq:supp-kinetic-constraint}
\mathbf F_i^{\mathrm c} =T_i\widehat{\mathbf b}_i -T_{i-1}\widehat{\mathbf b}_{i-1}.
\end{equation}
The tensions $T_i$ enforce the fixed length constraints.
The Brownian force $\mathbf F_i^{\mathrm B}$ has zero mean and satisfies the fluctuation-dissipation relation
\begin{equation}\label{eq:supp-kinetic-noise}
\langle\mathbf F_i^{\mathrm B}(t)\rangle=0,
\quad
\langle\mathbf F_i^{\mathrm B}(t)\mathbf F_j^{\mathrm B}(t')^{\top}\rangle=2\zeta k_{\mathrm B}T\delta_{ij}\mathbf I\delta(t-t').
\end{equation}

We integrate the force balance equation with a predictor-corrector scheme \cite{liu1989flexible} and solve the nonlinear bond-tension equations by Newton's method at each step \cite{somasi2002brownian}.
We use a time step $\Delta t=5\times10^{-4}\tau_d$, where $\tau_d=b^2\zeta/(k_{\mathrm B}T)$; the dimensionless contact repulsion strength is $\mu b/(k_{\mathrm B}T)=85.0$. Each chain is equilibrated for $10^4\tau_d$ (without any flow) before turning on the flow.  A single-exponential fit to the final $30\%$ extension during relaxation of an initially
stretched chain yields the longest relaxation time
$\tau=4321\tau_d$, defining the Weissenberg number $\mathrm{Wi}=\dot\epsilon\tau$. 
In simulation units, $\widetilde{\mathbf r}=\mathbf r/b$, $\widetilde t=t/\tau_d$, and $\widetilde{\mathbf F}=b\mathbf F/(k_{\mathrm B}T)$, such that
\begin{equation}\label{eq:supp-wi-simulation-units}
\widetilde{\dot\epsilon}=\dot\epsilon\tau_d,
\qquad \mathrm{Wi}=\widetilde{\dot\epsilon}\frac{\tau}{\tau_d} =4321\widetilde{\dot\epsilon}.
\end{equation}

\subsection{Symmetry-adapted dimensionality reduction}
\label{supp:kinetic-coordinates}
We represent each conformation by the stacked bead positions
$\mathbf r=(\mathbf r_0^\top,\ldots,\mathbf r_N^\top)^\top\in
\mathbb R^{900}$, retaining the instantaneous center-of-mass position.
The first reduced coordinate is the normalized chain extension,
\begin{equation}
 Z_1=\frac{\max_i r_{i,x}-\min_i r_{i,x}}{(N+1)b}.
 \label{eq:supp-kinetic-extension}
\end{equation} 
To construct the two additional coordinates, we sample $K$ training conformations across chain-extension bins within each batch, supplemented by random draws without replacement. 

We constrain the projection directions by chain reversal $R$ and coordinate reflections $T_x,T_y,T_z$. These operations commute and each squares to the identity. Their $16$ combinations form the symmetry group $G=\{R^{a_R}T_x^{a_x}T_y^{a_y}T_z^{a_z}: a_R,a_x,a_y,a_z\in\{0,1\}\}$.
For each $g\in G$, an orthogonal matrix $U_g\in\mathbb R^{900\times900}$ represents its action on $\mathbf r$. 
Specifically, $(U_R\mathbf r)_i=\mathbf r_{N-i}$, where $i$ labels a three-component bead block, and $U_{T_\alpha}$ is diagonal with $-1$ on the $\alpha$-coordinate entries of all beads and $+1$ elsewhere ($\alpha=x,y,z$). Composed operations satisfy $U_{gh}=U_gU_h$ for $g,h\in G$.
A one-dimensional character $\chi:G\to\{+1,-1\}$ satisfies $\chi(gh)=\chi(g)\chi(h)$ and is labeled by the generator signs $(\chi(R),\chi(T_x),\chi(T_y),\chi(T_z))$.
Each character defines a symmetry subspace of projection vectors $\mathbf v\in\mathbb R^{900}$ satisfying $U_g\mathbf v=\chi(g)\mathbf v$ for all $g\in G$, where $\chi(g)=+1$ or $-1$ indicates that $g$ preserves or reverses $\mathbf v$, respectively. The orthogonal projector onto this subspace is
\begin{equation}\label{eq:supp-kinetic-projector}
 P_\chi=\frac1{16}\sum_{g\in G}\chi(g)U_g.
\end{equation}
We select $\chi_2=(-1,-1,+1,+1)$ and $\chi_3=(+1,-1,+1,+1)$ for the two PCA directions.

The group-averaged sample mean vanishes under the coordinate
reflections. Accordingly, for the sampled conformations
$\mathbf r^{(m)}$, $m=1,\ldots,K$, we form the second-moment matrices
\begin{equation}
 C_0=\frac1{K-1}\sum_{m=1}^{K}\mathbf r^{(m)}
     \mathbf r^{(m)\top},
 \qquad
 C=\frac1{16}\sum_{g\in G}U_gC_0U_g^\top.
 \label{eq:supp-kinetic-covariance}
\end{equation}
This enforces a zero mean by symmetry without removing the center of mass from individual conformations.
Within each subspace, the unit vector $\mathbf v_j$ maximizing $\mathbf v^\top C\mathbf v$ is obtained as the normalized leading eigenvector of $P_{\chi_j}CP_{\chi_j}$.
For the projections $a_j^{(m)}=\mathbf v_j^\top\mathbf r^{(m)}$ on the original fitting samples, define
\begin{equation}\label{eq:supp-kinetic-pca}
\bar a_j=\frac1K\sum_m a_j^{(m)},\quad
s_j^2=\frac1K\sum_m(a_j^{(m)}-\bar a_j)^2,\quad
Z_j=\frac{\mathbf v_j^\top\mathbf r}{s_j},\quad j=2,3.
\end{equation}
Thus, each coordinate is scaled by its sample standard deviation without mean subtraction.
The resulting map $\mathbf r\mapsto\mathbf Z=(Z_1,Z_2,Z_3)$, including direction signs and scales, is fixed across all flow strengths and training and test trajectories.

\section{Conditional UIDD: formulation, learning, and landscape analysis}
\label{supp:kinetic-uidd}
\makeatletter
\ifcsname r@fig:2\endcsname
\expandafter\let\expandafter\reviewfigtwo\csname r@fig:2\endcsname
\expandafter\let\csname r@fig:com\endcsname\reviewfigtwo
\fi
\makeatother

\subsection{Expressivity and identifiability}
\label{supp:cuidd-representation}
We next establish the expressivity and identifiability of conditional UIDD. Specifically, we show that any diffusion family satisfying the assumptions below admits a conditional UIDD representation whose potential is unique up to a $c$-dependent additive constant. Consider a family of It\^{o} SDEs on $\mathbb R^d$ ($d\ge1$) indexed by a condition $c\in\mathcal U\subset\mathbb R^p$:
\begin{equation}
 \mathrm d\mathbf Z_t=b(\mathbf Z_t;c)\mathrm dt
          +\sigma(\mathbf Z_t;c)\mathrm d\mathbf B_t.
 \label{supp:cuidd-target-sde}
\end{equation}
Here $\mathbf B_t$ is a standard $d$-dimensional Brownian motion,
$b(\mathbf Z;c)\in\mathbb R^d$ is the drift, and
$\sigma(\mathbf Z;c)\in\mathbb R^{d\times d}$ is the noise coefficient. The condition domain $\mathcal U$ is an open set containing the
conditions $\mathcal C$, and $c$ is fixed
along each trajectory. 
We assume a positive-definite diffusion matrix $M=\sigma\sigma^\top/2$.
We write $\mathbf e_i$ for the standard basis vectors of $\mathbb R^d$.

\newtheorem{cuiddassumption}{Assumption}
\renewcommand{\thecuiddassumption}{S\arabic{cuiddassumption}}
\newtheorem{cuiddtheorem}{Theorem}
\renewcommand{\thecuiddtheorem}{S\arabic{cuiddtheorem}}
\newtheorem{cuiddlemma}{Lemma}
\renewcommand{\thecuiddlemma}{S\arabic{cuiddlemma}}

\begin{cuiddassumption}
\label{supp:cuidd-assumption-dynamics}
For each fixed $c\in\mathcal U$, the coefficients $b(\cdot;c)\in C^1$ and $\sigma(\cdot;c)\in C^2$ of Eq.~\eqref{supp:cuidd-target-sde} satisfy global linear growth and Lipschitz bounds in $\mathbf Z$ (with $c$-dependent constants),
and the process admits a unique invariant probability density $\pi(\cdot;c)\in C^2\cap L^1$. 
\end{cuiddassumption}

\begin{cuiddassumption}
\label{supp:cuidd-assumption-density}
For all $i,j=1,\ldots,d$, assume
\[
\begin{aligned}
 b,\ \partial_{Z_i}b
  &\in C^0(\mathbb R^d\times\mathcal U),\\
 M,\ \partial_{Z_i}M,\ \partial_{Z_i}\partial_{Z_j}M
  &\in C^0(\mathbb R^d\times\mathcal U),\\
 \pi,\ \partial_{Z_i}\pi,\ \partial_{Z_i}\partial_{Z_j}\pi
  &\in C^0(\mathbb R^d\times\mathcal U).
\end{aligned}
\]
Here continuity is joint in $(\mathbf Z,c)$ and is understood componentwise for vector and matrix fields.
\end{cuiddassumption}

For Eq.~\eqref{supp:cuidd-target-sde}, the stationary Fokker-Planck equation requires a divergence-free probability current:
\begin{equation}
 \mathbf J=b\pi-\nabla_{\mathbf Z}\cdot(M\pi),
 \qquad \nabla_{\mathbf Z}\cdot\mathbf J=0.
 \label{supp:cuidd-target-current}
\end{equation}
Here $\mathbf J(\cdot;c)\in C^1$ for each $c$. Under Assumption~\ref{supp:cuidd-assumption-density}, $\mathbf J$ and its first $\mathbf Z$-derivatives are jointly continuous in $(\mathbf Z,c)$.

\begin{cuiddlemma}
\label{supp:cuidd-banded-lemma}
Let $d\ge2$ and $\mathbf J:\mathbb R^d\times\mathcal U\to\mathbb R^d$ satisfy $\nabla_{\mathbf Z}\cdot\mathbf J=0$ and
$\mathbf J,\partial_{Z_i}\mathbf J\in C^0(\mathbb R^d\times\mathcal U)$ for $i=1,\ldots,d$.
There exists a skew-symmetric matrix field $\Psi:\mathbb R^d\times\mathcal U\to\mathbb R^{d\times d}$, with $\Psi_{ij}=0$ for $|i-j|>1$, such that $\nabla_{\mathbf Z}\cdot\Psi=\mathbf J$ and
$\Psi,\partial_{Z_i}\Psi_{i,i+1},\partial_{Z_{i+1}}\Psi_{i,i+1}\in C^0(\mathbb R^d\times\mathcal U)$ for $i=1,\ldots,d-1$.
\end{cuiddlemma}

\begin{proof}
We adapt the coordinate-integration construction in Ref.~\cite{kang1995volume}, using the reference hyperplanes $Z_j=0$ for every $c$.
For a scalar field $f=f(\mathbf Z;c)$, define
\[
\begin{aligned}
 (I_jf)(\mathbf Z;c) =\int_0^{Z_j}f(Z_1,\ldots,s,\ldots,Z_d;c)\mathrm ds,\quad 
 (E_jf)(\mathbf Z;c) =f(Z_1,\ldots,0,\ldots,Z_d;c),
\end{aligned}
\]
where $s$ or $0$ replaces the $j$th coordinate of $\mathbf Z$ at fixed $c$.
Both operators preserve joint continuity in $(\mathbf Z,c)$, since $E_j$ is a coordinate evaluation and
\[
 (I_jf)(\mathbf Z;c)
 =Z_j\int_0^1f(Z_1,\ldots,uZ_j,\ldots,Z_d;c)\mathrm du 
\]
is continuous in $(\mathbf Z,c)$ by uniform continuity on compact sets.
Define $S_i(\mathbf Z;c)=\sum_{k=1}^i\partial_{Z_k}J_k(\mathbf Z;c)$ and $S_0=0$.
For $d=2$, set
\begin{equation}
 A_1(\mathbf Z;c)=(I_2J_1-I_1E_2J_2)(\mathbf Z;c).
 \label{supp:cuidd-current-construction}
\end{equation}
At each $c$, the divergence-free condition $\nabla_{\mathbf Z}\!\cdot\mathbf J=0$ implies
$\partial_{Z_2}A_1(\mathbf Z;c)=J_1(\mathbf Z;c)$ and
$\partial_{Z_1}A_1(\mathbf Z;c)=-J_2(\mathbf Z;c)$.
For $d\ge3$, we take
\begin{equation}
 \begin{aligned}
 A_1(\mathbf Z;c)&=(I_2J_1)(\mathbf Z;c),\\
 A_i(\mathbf Z;c)&=(I_{i+1}J_i+I_{i+1}I_iS_{i-1})(\mathbf Z;c),
 \quad 2\le i\le d-2,\\
 A_{d-1}(\mathbf Z;c)&=(I_dE_{d-1}J_{d-1}-I_{d-1}J_d)(\mathbf Z;c).
 \end{aligned}
 \label{supp:cuidd-induction}
\end{equation}
Differentiating at fixed $c$ and using $\partial_{Z_i}(I_iS_{i-1})=S_{i-1}$, we obtain that, for $1\le i\le d-2$,
\[
\begin{aligned}
 \partial_{Z_i}A_i(\mathbf Z;c)
 &=[I_{i+1}(\partial_{Z_i}J_i+S_{i-1})](\mathbf Z;c)
  =(I_{i+1}S_i)(\mathbf Z;c),\\
 \partial_{Z_{i+1}}A_i(\mathbf Z;c)
 &=J_i(\mathbf Z;c)+(I_iS_{i-1})(\mathbf Z;c).
\end{aligned}
\]
For the last entry, $\nabla_{\mathbf Z}\cdot\mathbf J(\mathbf Z;c)=0$ yields
\[
\begin{aligned}
 \partial_{Z_{d-1}}A_{d-1}(\mathbf Z;c) =-J_d(\mathbf Z;c),\quad
 \partial_{Z_d}A_{d-1}(\mathbf Z;c)
  =J_{d-1}(\mathbf Z;c)+(I_{d-1}S_{d-2})(\mathbf Z;c).
\end{aligned}
\]
Set $\Psi_{i,i+1}(\mathbf Z;c)=A_i(\mathbf Z;c)=-\Psi_{i+1,i}(\mathbf Z;c)$ and all other entries to zero. Then
\begin{equation}
 \begin{aligned}
 (\nabla_{\mathbf Z}\cdot\Psi)_1(\mathbf Z;c)&=J_1(\mathbf Z;c),\\
 (\nabla_{\mathbf Z}\cdot\Psi)_i(\mathbf Z;c)
 &=-\partial_{Z_{i-1}}A_{i-1}(\mathbf Z;c)
  +\partial_{Z_{i+1}}A_i(\mathbf Z;c)=J_i(\mathbf Z;c),
 \quad 2\le i\le d-1,\\
 (\nabla_{\mathbf Z}\cdot\Psi)_d(\mathbf Z;c)
 &=-\partial_{Z_{d-1}}A_{d-1}(\mathbf Z;c)=J_d(\mathbf Z;c).
 \end{aligned}
 \label{supp:cuidd-current-verification}
\end{equation}
Since $J_i$ and $S_i$ are jointly continuous in $(\mathbf Z,c)$ and $I_j,E_j$ preserve continuity, $\Psi$ and its stated derivatives are jointly continuous.
The construction requires no $c$-derivatives and needs only $\mathbf J(\cdot;c)\in C^1$ at each fixed $c$.
For $d=2,3$, all first $\mathbf Z$-derivatives of $\Psi$ are jointly continuous; for $d\ge4$, this holds if all second $\mathbf Z$-derivatives of $\mathbf J$ are jointly continuous.
\end{proof}

\begin{cuiddtheorem}[{Expressivity and identifiability}]
\label{supp:cuidd-representation-theorem}
For the family of SDEs in Eq.~\eqref{supp:cuidd-target-sde} satisfying Assumption~\ref{supp:cuidd-assumption-dynamics}, there exist a potential $V(\cdot;c)\in C^2(\mathbb R^d)$ and an antisymmetric matrix $W(\mathbf Z;c)$, with $W_{ij}=0$ for $|i-j|>1$, such that $e^{-V(\cdot;c)}\in L^1(\mathbb R^d)$ and
\begin{equation}
 b(\mathbf Z;c)=-[M(\mathbf Z;c)+W(\mathbf Z;c)]\nabla_{\mathbf Z}V(\mathbf Z;c)
     +\nabla_{\mathbf Z}\cdot M(\mathbf Z;c)+\nabla_{\mathbf Z}\cdot W(\mathbf Z;c)
 \label{supp:cuidd-exact-representation}
\end{equation}
for every $(\mathbf Z,c)\in\mathbb R^d\times\mathcal U$, with $W=0$ for $d=1$.
For each fixed $c$, the entries $H_i=W_{i,i+1}$ and their derivatives $\partial_{Z_i}H_i$, $\partial_{Z_{i+1}}H_i$ are continuous in $\mathbf Z$.
The potential can be chosen as $V=-\ln\pi$. Among representations with continuous potentials whose normalized exponentials are invariant for Eq.~\eqref{supp:cuidd-target-sde}, it is unique up to an additive function $C(c)$ independent of $\mathbf Z$.

If Assumption~\ref{supp:cuidd-assumption-density} also holds, $V$ (up to second $\mathbf Z$-derivatives) and $H_i,\partial_{Z_i}H_i,\partial_{Z_{i+1}}H_i$ can be chosen jointly continuous in $(\mathbf Z,c)$.
\end{cuiddtheorem}

\begin{proof}
We follow the UIDD decomposition~\cite{zhu2026identifiable}, first fixing $c$.
On any bounded ball, the stationary Fokker-Planck operator is uniformly elliptic. Bounding its zeroth-order coefficient below by a nonpositive constant makes $\pi\ge0$ a supersolution, so the strong minimum principle and normalization imply $\pi>0$ on $\mathbb R^d$.
Using the stationary current in Eq.~\eqref{supp:cuidd-target-current}, decompose the drift as
\[
 b=\frac{\nabla_{\mathbf Z}\cdot(M\pi)}{\pi}
  +\frac{\mathbf J}{\pi},
 \qquad \nabla_{\mathbf Z}\cdot\mathbf J=0.
\] 
For $d\ge2$, the construction in Lemma~\ref{supp:cuidd-banded-lemma} yields a banded antisymmetric matrix $\Psi$ with $\nabla_{\mathbf Z}\cdot\Psi=\mathbf J$.
Define
\[
 V(\mathbf Z;c)=-\ln\pi(\mathbf Z;c),
 \qquad W(\mathbf Z;c)=\frac{\Psi(\mathbf Z;c)}{\pi(\mathbf Z;c)}.
\]
The matrix $W$ has the same banded antisymmetric structure as $\Psi$. Since $\nabla_{\mathbf Z}\pi=-\pi\nabla_{\mathbf Z}V$, the two terms in the drift become
\begin{equation}
 \begin{aligned}
 \frac{\nabla_{\mathbf Z}\cdot(M\pi)}{\pi}
   =-M\nabla_{\mathbf Z}V+\nabla_{\mathbf Z}\cdot M,\quad
 \frac{\mathbf J}{\pi}
   =\frac{\nabla_{\mathbf Z}\cdot(W\pi)}{\pi}
    =-W\nabla_{\mathbf Z}V+\nabla_{\mathbf Z}\cdot W.
 \end{aligned}
 \label{supp:cuidd-representation-verification}
\end{equation}
Adding these identities yields Eq.~\eqref{supp:cuidd-exact-representation}, and $e^{-V}=\pi$ is normalized and invariant.

Under Assumption~\ref{supp:cuidd-assumption-density}, $\mathbf J$ and its first $\mathbf Z$-derivatives are jointly continuous, so Lemma~\ref{supp:cuidd-banded-lemma} yields a $\Psi(\mathbf Z;c)$ with the stated joint continuity. Dividing by $\pi>0$ and applying the chain rule gives the claimed regularity of $W$ and $V$ without differentiating in $c$.

Finally, if $\widetilde V$ is another continuous potential whose normalized exponential is invariant, uniqueness of the continuous invariant density implies
\[
 \frac{e^{-\widetilde V(\mathbf Z;c)}}{\widetilde{\mathcal Z}(c)}
 =\pi(\mathbf Z;c),
 \quad
 \widetilde{\mathcal Z}(c)=\int_{\mathbb R^d}e^{-\widetilde V(\mathbf Z;c)}\mathrm d\mathbf Z.
\]
Thus $\widetilde V(\mathbf Z;c)=-\ln\pi(\mathbf Z;c)+C(c)$ with $C(c)=-\ln\widetilde{\mathcal Z}(c)$, proving identifiability.
\end{proof}

\subsection{Neural parameterization}
\label{supp:cuidd-approximation}
We extend the UIDD network modules~\cite{zhu2026identifiable} by supplying the joint input $(\mathbf Z,c)$ and sharing all parameters $\theta$ across conditions:
\begin{equation}
\begin{aligned}
 \sigma_\theta(\mathbf Z;c)
 =&\sqrt2\left[\alpha_{\sigma} T_\theta(\mathbf Z;c)+\operatorname{diag}\phi(s_\theta(\mathbf Z;c))\right],\\
 W_\theta(\mathbf Z;c) =& \begin{pmatrix}
  0& H_{\theta,1}(\mathbf Z;c)&0 &\cdots&0&0\\
  -H_{\theta,1}(\mathbf Z;c)&0&H_{\theta,2}(\mathbf Z;c)&\cdots&0&0\\
  0&-H_{\theta,2}(\mathbf Z;c)&0&\cdots&0&0\\
  \vdots&\vdots&\vdots& \ddots&\vdots&\vdots\\
  0&0&0&\cdots&0&H_{\theta, d-1}(\mathbf Z;c)\\
  0&0&0&\cdots&-H_{\theta, d-1}(\mathbf Z;c)&0\\
\end{pmatrix},
\end{aligned}
\label{eq:supp-kinetic-uidd-w}
\end{equation} where $T_\theta\in\mathbb R^{d\times d}$ is strictly lower triangular, $s_\theta\in\mathbb R^d$, $H_\theta\in\mathbb R^{d-1}$, and $\phi(s)=(\sqrt{1+s^2}+s)/2$ is applied componentwise. Thus $M_\theta=\sigma_\theta\sigma_\theta^\top/2=L_\theta L_\theta^\top$, where $L_\theta=\alpha_{\sigma}T_\theta+\operatorname{diag}\phi(s_\theta)$ has positive diagonal. The off-diagonal scale factor is fixed at $\alpha_{\sigma}=0.1$. The potential is
\begin{equation}
 V_\theta(\mathbf Z;c)
 =\frac{\beta_1}{2}\|U_\theta(\mathbf Z;c)+\beta_2\Gamma(\mathbf Z^\top,c^\top)^\top\|^2
 +\beta_3\|\mathbf Z\|^2,
 \label{supp:cuidd-admissible-potential}
\end{equation}
with $U_\theta\in\mathbb R^m$, $m\ge1$, trainable $\Gamma\in\mathbb R^{m\times(d+p)}$, and fixed $\beta_1,\beta_2,\beta_3>0$.
Thus $M_\theta$ is positive definite, $W_\theta$ is banded and antisymmetric, and $e^{-V_\theta}\le e^{-\beta_3\|\mathbf Z\|^2}$ is integrable.

To establish expressivity, consider $\tanh$ networks. Under Assumptions~\ref{supp:cuidd-assumption-dynamics} and \ref{supp:cuidd-assumption-density}, let $V,W$ be the fields in Theorem~\ref{supp:cuidd-representation-theorem} and let $Q=K\times\mathcal C_0$, where $K\subset\mathbb R^d$ and $\mathcal C_0\subset\mathcal U$ are nonempty compact sets. For every $\epsilon>0$, there exist shared parameters $\theta$ and a constant $C_Q$, independent of $c$, such that
\begin{equation}
 \sup_{(\mathbf Z,c)\in Q}
 \left[|V_\theta-V-C_Q|+\|f_\theta-b\|+\|M_\theta-M\|\right]<\epsilon,
 \label{supp:cuidd-generator-approximation}
\end{equation}
where $f_\theta$ is the drift in Eq.~\eqref{eq:uidd-main} assembled from the neural fields.

To see this, let \(M=LL^\top\) be the Cholesky factorization, with
\(L\) lower triangular and \(L_{ii}>0\). We approximate \(10L_{ij}\)
for \(i>j\) and \(\phi^{-1}(L_{ii})\) by the corresponding components
of \(T_\theta\) and \(s_\theta\), respectively. For the potential, we
choose \(C_Q\) such that
\(V+C_Q-\beta_3\|\mathbf Z\|^2>0\) on a neighborhood of
\(Q\), set \(\Gamma=0\), and approximate
\[
q(\mathbf Z;c)
=
\sqrt{\frac{2}{\beta_1}
\left[V(\mathbf Z;c)+C_Q-\beta_3\|\mathbf Z\|^2\right]}
\]
by one component of \(U_\theta\), with the remaining components set to zero. 
We approximate these targets together with their required $\mathbf Z$-derivatives (in particular, $H_i=W_{i,i+1}$ along with $\partial_{Z_i}H_i$ and $\partial_{Z_{i+1}}H_i$).
By Assumption~\ref{supp:cuidd-assumption-density} and Lemma~\ref{supp:cuidd-banded-lemma}, all targets and required derivatives are jointly continuous in $(\mathbf Z,c)$.
Multiplying by a smooth cutoff equal to one near $Q$ and mollifying in $(\mathbf Z,c)$ yields smooth functions that converge uniformly on $Q$ together with these derivatives.
Simultaneous neural approximation of functions and derivatives~\cite{hornik1990derivatives}, combined with the product and chain rules, then yields Eq.~\eqref{supp:cuidd-generator-approximation}.

\emph{Polymer implementation.}
We symmetrize the potential as
\begin{equation}
 V_\theta^{\mathrm{sym}}(\mathbf Z;c)
 =\frac14\sum_{g\in G_Z}V_\theta(g\mathbf Z;c),
 \qquad G_Z=\{\operatorname{diag}(1,\pm1,\pm1)\}.
 \label{supp:cuidd-symmetric-potential}
\end{equation}
The generators $\operatorname{diag}(1,-1,1)$ and $\operatorname{diag}(1,-1,-1)$ are the reduced actions of $R$ and $T_x$, respectively (Sec.~\ref{supp:kinetic-coordinates}).
Thus the potential is even in $Z_2$ and $Z_3$ and retains the Gaussian bound; $\sigma_\theta$ and $W_\theta$ are not symmetrized.
For a $G_Z$-invariant target potential and a $G_Z$-invariant $K$, averaging does not increase the uniform errors in the potential and its state gradient. The preceding approximation argument therefore also applies to this symmetry-restricted class.

\begin{samepage}
Network hyperparameters are summarized in the table below. 
\begin{center}
\begin{tabular}{llc}
\hline
Output & Layer widths & Hidden activation \\
\hline
$U_\theta$ & $4\to128\to32$ & $r$ \\
$s_\theta$ & $4\to32\to32\to3$ & $r$ \\
$T_\theta$ (before triangular projection) & $4\to32\to32\to9$ & $r$ \\
$H_\theta$ & $4\to32\to32\to2$ & $\tanh$ \\
\hline
\end{tabular}
\end{center}
\end{samepage}
Here $r(x)=\max(0,x)^2-\max(0,x-\tfrac12)^2$.
For $T_\theta$, the nine outputs are reshaped into a $3\times3$ matrix before retaining its strictly lower-triangular part. Here $\Gamma$ is a bias-free linear layer with $m=32$, and $(\beta_1,\beta_2,\beta_3)=(100,1,0.01)$.

\subsection{Training and testing}
\label{supp:kinetic-uidd-learning}
We fit the conditional model to \(11145\) trajectories by minimizing the loss in Eq.~\eqref{eq:em-uidd-loss}.
Of these, \(4945\) trajectories start from a pool of coiled conformations and another \(4945\) from a single folded conformation.
Both ensembles span \(\mathrm{Wi}\in(0.005,10.8]\) with closely matched condition distributions.
An additional \(1255\) trajectories sample \(\mathrm{Wi}\in(0,1)\) from the same pool of coiled conformations.
The condition is held fixed along each trajectory.

Training uses adjacent transitions from \(500\)-frame segments with
\(\Delta t_{\mathrm{train}}=5\times10^{-4}\).
Minibatches contain at most \(256\) segments and are processed in a fixed order.
For each minibatch, a single transformation is drawn from the identity, \(R\), and \(T_x\) with probabilities \(1/2\), \(1/4\), and \(1/4\), respectively.
The model is trained for \(300\) epochs with Adam, starting from a learning rate of \(10^{-3}\), which is reduced when the loss plateaus.
We retain the parameters from the final epoch and train five independent models from different random initializations.

We evaluate transient stretching from three initial microscopic conformations (labeled fast, medium, and slow) excluded from both PCA fitting and model training.
For each conformation and each of ten test $\mathrm{Wi}$ values, we compare $500$ microscopic trajectories against $500$ reduced trajectories simulated from the corresponding encoded initial state.
To quantify unfolding times, let $Z_1^{(j)}(t)$ denote the normalized span of trajectory $j$, and define the unfolding step $n_{\mathrm u}^{(j)}$ as
\begin{equation}
 n_{\mathrm u}^{(j)}
 =\min\left\{n\in\{0,\ldots,n_{\max}\}:
 Z_1^{(j)}(n\Delta t_{\mathrm{test}})
 >0.95Z_1^{(j)}(t_{\max})\right\},
 \label{supp:cuidd-unfolding-time}
\end{equation}
where \(n_{\max}=t_{\max}/\Delta t_{\mathrm{test}}\). 
This threshold is defined relative to each trajectory's final extension and is recorded at discrete output steps without interpolation ($n_{\mathrm u}^{(j)}=0$ if exceeded initially).
Figure~\ref{fig:supp-kinetic-transient-tests} shows that the conditional UIDD model accurately reproduces the ensemble mean and standard deviation of $Z_1(t)$ as well as the normalized histograms of $n_{\mathrm u}$ across all three initial conformations and flow strengths.

\begin{figure}[htbp]
\centering
\includegraphics[width=\linewidth,height=0.82\textheight,keepaspectratio]{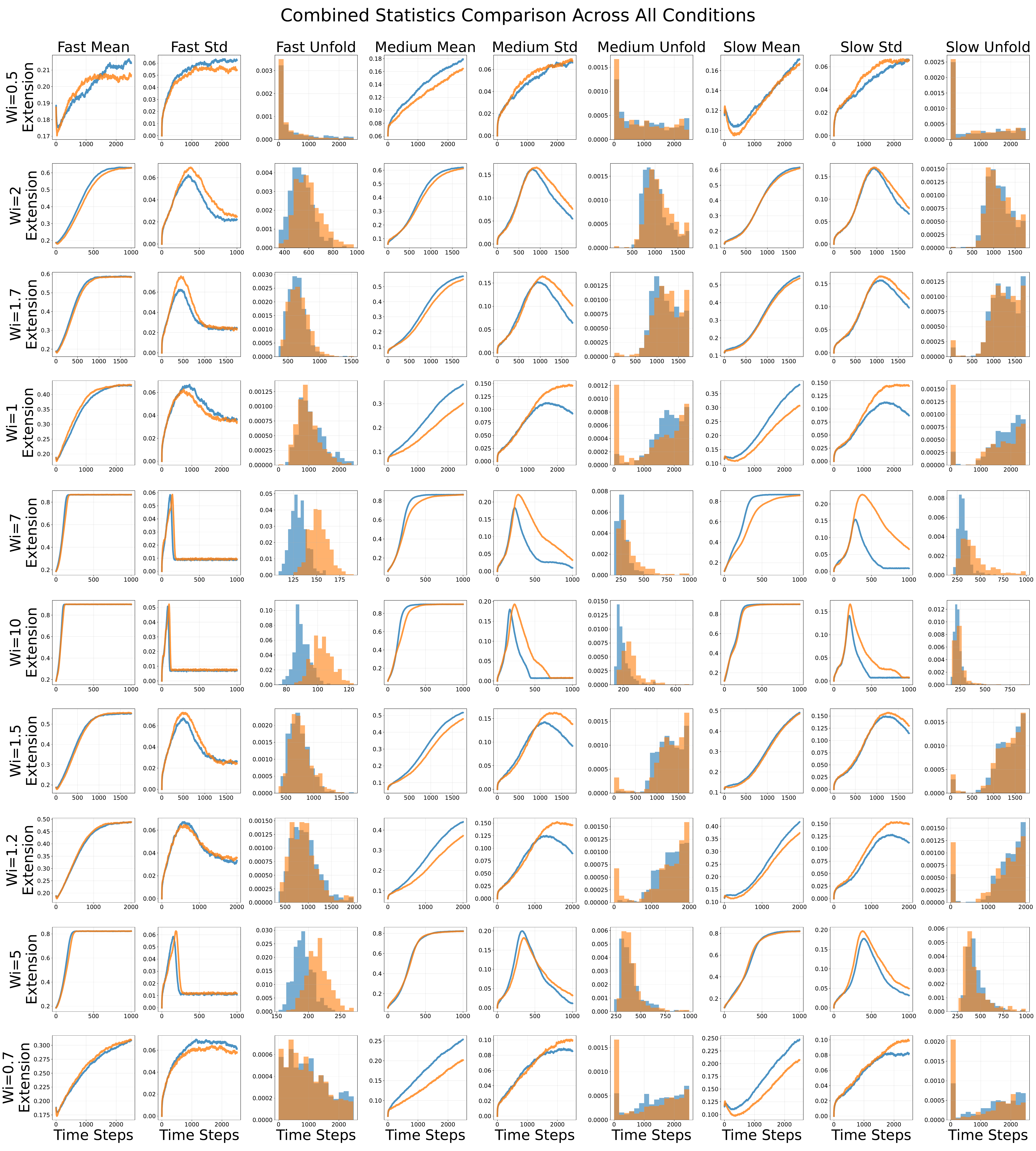}
\caption{Transient stretching from the fast, medium, and slow initial conformations (column groups) at different \(\mathrm{Wi}\) (rows). Each group compares the mean and standard deviation of \(Z_1\) and the distribution of \(n_{\mathrm u}\) defined in Eq.~\eqref{supp:cuidd-unfolding-time} obtained from microscopic simulations (blue) and the conditional UIDD model (orange), using \(500\) realizations per ensemble. Horizontal axes are given in output-step units.}
\label{fig:supp-kinetic-transient-tests}
\end{figure}

\subsection{Landscape evaluation and transition estimates}
\label{supp:kinetic-transitions}
We evaluate the marginal integrals using uniform-grid Riemann sums (via log-sum-exp) with equal grid-spacing weights on the following grids:
\begin{center}
\begin{tabular}{lll}
\hline
Profile & Integration grids (nodes) & Tabulation grid (nodes; spacing) \\
\hline
$F_2$ &\shortstack[l]{$Z_1\in[-0.5,2.5]$ ($512$),\\$Z_3\in[-15,15]$ ($769$)}
   &$Z_2\in[-3,3]$ ($601$; $0.01$) \\
$F_{23}(0,\cdot)$ & $Z_1\in[-0.5,2.5]$ ($512$)
   & $Z_3\in[-15,15]$ ($2401$; $0.0125$) \\
\hline
\end{tabular}
\end{center} 
All profiles and insets use the same $512$-node quadrature
over $Z_1\in[-0.5,2.5]$.

For $k=1,\ldots,5$, let $\mathcal F_k(q;\mathrm{Wi})$ denote either $F_2^{(k)}(q;\mathrm{Wi})$ or $F_{23}^{(k)}(0,q;\mathrm{Wi})$, and let $\mathcal Q$ be the corresponding tabulation grid, which includes zero. For $\eta>0$, define
\begin{align}
\mathcal A_{\eta,k}(\mathrm{Wi})&=\left\{q\in\mathcal Q:\mathcal F_k(q;\mathrm{Wi})-\min_{q'\in\mathcal Q}\mathcal F_k(q';\mathrm{Wi})\le\eta\right\},\\
 D_k(\mathrm{Wi})
 &=\mathcal F_k(0;\mathrm{Wi})
 -\min_{q\in\mathcal Q}\mathcal F_k(q;\mathrm{Wi}).
\end{align}
The condition $0\notin\mathcal A_{\eta,k}$ is equivalent to $D_k>\eta$. 
Using $\eta=0.6$ for $F_2$ and $\eta=1.0$ for
$F_{23}(0,\cdot)$ [Fig.~\ref{fig:2}(i,l)], we estimate each model's transition point $\mathrm{Wi}_{t,k}$ (denoted $\mathrm{Wi}_t$ in the main text) by linearly interpolating $D_k(\mathrm{Wi})-\eta$ across its first upward zero crossing.
The $F_2$ scan contains $146$ values of $\mathrm{Wi}$: spacing $0.02$ on $[0.02,1]$, refined to $0.005$ on $[0.15,0.55]$; spacing $0.05$ on $[1.05,2]$; and $15$ logarithmically spaced values on $[2.2,10]$.
The $F_{23}(0,\cdot)$ scan uses $439$ values: spacing $0.01$ on $[0.02,3]$ and $0.05$ on $[3.05,10]$.
At finite $\eta$, $\mathrm{Wi}_{t,k}$ marks where the central barrier reaches $\eta$, rather than the exact onset of peak splitting.

The tolerance-sensitivity panels in Fig.~\ref{fig:2}(j,m)
show the individual model estimates, their mean, and error bars of
one sample standard deviation over five initializations, with
\[
\begin{aligned}
 F_2:\quad
 &\eta\in\{0.02,0.05,0.1,0.2,0.3,0.4,0.5,0.6,0.8,1.0,1.2,1.5,2.0\},\\
 F_{23}(0,\cdot):\quad
 &\eta\in\{0.05,0.1,0.2,0.3,0.4,0.6,0.8,1.0,1.2,1.5,2.0,3.0\}.
\end{aligned}
\]
The spread across models reflects initialization sensitivity on the fixed training set. For visibility in Fig.~\ref{fig:2}(i,l), near-minimizing regions are plotted with minimum widths of four $Z_2$ and nine $Z_3$ grid cells; all $\mathrm{Wi}_{t,k}$ estimates use the unbroadened sets.

The numerical quadratures are stable under grid refinement. Increasing the \(F_2\) quadrature to \(1023\) nodes in \(Z_1\) and \(1537\) nodes in \(Z_3\), and the section quadrature to \(2048\) nodes in \(Z_1\), changes \(D_k\) by less than \(10^{-3}\) at the tested conditions with \(\mathrm{Wi}\le3\) and by less than \(4\times10^{-2}\) at those up to \(\mathrm{Wi}=10\). All tabulated minima of \(F_{23}(0,Z_3)\) satisfy \(|Z_3|<9\).
To test domain truncation, we compare the integral $I_{\mathrm{win}}$ on the default window with $I_{\mathrm{ref}}$ on an expanded window ($Z_1\in[-1.50,3.50]$, $Z_3\in[-20,20]$) at matched spacings ($3/511$ in $Z_1$, $10/256$ in $Z_3$). Across all five models, evaluating $1-I_{\mathrm{win}}/I_{\mathrm{ref}}$ at $Z_2=0$ and at the grid minimizer for five representative controls (the point nearest $\mathrm{Wi}_{t,k}^{(1)}$, namely $0.385,0.455,0.435,0.370,0.460$, plus $\mathrm{Wi}=1,2.73,5.23,10$) yields values below $10^{-12}$, confirming negligible truncation error.

\section{Large-deviation analysis of the analytical model}
\label{supp:steady-model}

\newtheorem{steadylemma}[cuiddlemma]{Lemma}
\newtheorem{steadytheorem}[cuiddtheorem]{Theorem}

We prove the limit interchange in Eq.~\eqref{eq:em-f2-limit} and
characterize the minimizing set of the joint free energy.
We use the potential, observables, and reference measure of the
main-text analytical model, with
$\mathcal E$ and $\mathcal S$ defined in Eq.~\eqref{eq1},
and $\ell$ in Eq.~\eqref{eq:em-bond-entropy}.
For $\alpha=2,23$, write $d_\alpha=1,2$, respectively, and set
\[
 \hat Q_2=\hat z_2,\quad \hat Q_{23}=(\hat z_2,\hat z_3),\qquad
 Q_2(m)=m_x(1),\quad Q_{23}(m)=(m_x(1),m_x(1)-2\bar m_x).
\]
Let $\mathcal Z_{\alpha,n}$ be the unnormalized Boltzmann-weighted
density of $\hat Q_\alpha$, defined by the Dirac $\delta$ integrals in
the main text. We use its continuous version for all sufficiently
large $n$ and define
\begin{equation}
\begin{aligned}
 W_{\alpha,n,\varepsilon}(q;\lambda)
 =\int e^{-U_{n,\lambda}}
    \mathbf1_{\{|\hat Q_\alpha-q|_\infty<\varepsilon\}}\mathrm d\nu_n,\quad
 A_{\alpha,n,\varepsilon}(q;\lambda)
 =\frac{W_{\alpha,n,\varepsilon}(q;\lambda)}{(2\varepsilon)^{d_\alpha}}.
\end{aligned}
\label{supp:window-partition}
\end{equation}

Here $|\cdot|_\infty$ denotes the maximum norm, unsubscripted norms are Euclidean, and we set $\log0=-\infty$ and $\inf\varnothing=+\infty$. We now state our two main results.

\begin{steadytheorem}[Exchange of the two limits in Eq.~\eqref{eq:em-f2-limit}]
\label{supp:limit-interchange}
Fix $\lambda\ge0$. For $\alpha=2$, let $q=z_2\in(-1,1)$; for
$\alpha=23$, let $q=(0,z_3)$ with $|z_3|<1/2$.
The pointwise limits
 $F_{\alpha,\infty}(q;\lambda):=\lim_{n\to\infty}F_{\alpha,n}(q;\lambda)$,
where $F_{\alpha,n}=-n^{-1}\log\mathcal Z_{\alpha,n}$, exist and satisfy
\begin{equation}
\begin{aligned}
 F_{\alpha,\infty}(q;\lambda)
 =-\lim_{n\to\infty}\lim_{\varepsilon\downarrow0}
    \frac1n\log A_{\alpha,n,\varepsilon}(q;\lambda)
 =-\lim_{\varepsilon\downarrow0}\lim_{n\to\infty}
    \frac1n\log A_{\alpha,n,\varepsilon}(q;\lambda)
 =\inf_{\substack{m\in\mathcal A\\Q_\alpha(m)=q}}
    \{\lambda\mathcal E[m]-\mathcal S[m]\}.
\end{aligned}
\label{supp:iterated-limit-equality}
\end{equation}

\end{steadytheorem}

\begin{steadytheorem}[The minimizing set is an interval]
\label{supp:minimizing-interval-theorem}
For every finite $\lambda\ge0$, there is $a_\lambda\ge0$ such that
\begin{equation}
 \operatorname*{argmin}_{c\in\mathbb R}
 F_{23,\infty}(0,c;\lambda)=[-a_\lambda,a_\lambda].
 \label{eq:supp:transition-exact-interval}
\end{equation}
The interval opens at the exact threshold $\lambda_2=4\pi^2$. That is,
 $
 a_\lambda=0\ (0\le\lambda\le4\pi^2),\ 
 a_\lambda>0\ (\lambda>4\pi^2)$.
\end{steadytheorem}
The two-dimensional variational problem is able to be reduced to a scalar one.
Let $\bar u=\int_0^1u(s)\mathrm ds$, and take
\begin{equation}
\begin{aligned}
\mathcal A_1=\{u\in\mathrm{AC}[0,1]:u(0)=0,\ |u'|\le1\ \text{a.e.}\},\quad
\mathcal F^{(1)}_\lambda[u]
=\int_0^1\ell(u')\mathrm ds
-\lambda\int_0^1(u-\bar u)^2\mathrm ds.
\end{aligned}
\label{supp:scalar-functional}
\end{equation}
Since $\ell$ is even, convex, and nondecreasing on $[0,1]$ and $\lambda\ge0$, projecting $m=(u,v)$ onto $(u,0)$ preserves the constraints, drops the nonnegative transverse energy, and does not increase the entropy cost:
\begin{equation}
\label{supp:scalar-projection}
\begin{aligned}
\lambda\mathcal E[(u,v)]-\mathcal S[(u,v)]
\ge\lambda\mathcal E[(u,0)]-\mathcal S[(u,0)]
=\mathcal F^{(1)}_\lambda[u].
\end{aligned}
\end{equation}

Theorem~\ref{supp:limit-interchange} therefore yields the scalar variational formulas used in the proof of Theorem~\ref{supp:minimizing-interval-theorem} and in the numerical evaluation:
\begin{equation}
\begin{aligned}
F_{2,\infty}(z_2;\lambda)
=\inf_{\substack{u\in\mathcal A_1\\u(1)=z_2}}
\mathcal F^{(1)}_\lambda[u],\quad
F_{23,\infty}(0,z_3;\lambda)
=\inf_{\substack{u\in\mathcal A_1\\u(1)=0,\ \bar u=-z_3/2}}
\mathcal F^{(1)}_\lambda[u].
\end{aligned}
\label{supp:free-energy-limits}
\end{equation}

\subsection{Proof of the limit-interchange theorem}
\label{supp:steady-density}

\paragraph{Window limit.}
We use the path large-deviation principle and uniform energy approximation from the End Matter. With $\bar m_{n,x}=\int_0^1 m_{n,x}(s)\mathrm ds$ denoting the continuum average of the piecewise-linear path $m_n$, trapezoidal integration yields that the observables satisfy $\hat z_2=m_{n,x}(1)$ and
\begin{equation}
 \hat z_3=\frac n{n+1}\bigl[m_{n,x}(1)-2\bar m_{n,x}\bigr]
 =\frac1n\sum_{i=1}^n\frac{2i-n-1}{n+1}\cos\theta_i.
 \label{supp:continuum-identities}
\end{equation}
Thus $\hat Q_{23}=Q_{23}(m_n)+O(n^{-1})$ uniformly, so windows of widths $\varepsilon-\eta$ and $\varepsilon+\eta$ give the constrained Laplace bounds for any fixed $0<\eta<\varepsilon$.

For $q$ as in Theorem~\ref{supp:limit-interchange}, define
\[
 m_{\rm ref}(s)=\begin{cases}
  (z_2s,0),&\alpha=2,\quad q=z_2,\\
  (-2z_3\min(s,1-s),0),&\alpha=23,\quad q=(0,z_3).
 \end{cases}
\]
Each path lies in $\mathcal A$ and satisfies $Q_\alpha(m_{\rm ref})=q$. Its $y$-component is zero, so $\mathcal E[m_{\rm ref}]\le0$. The derivative is
\[
 m_{\rm ref}'(s)=\begin{cases}
  (z_2,0),&\alpha=2,\\
  (-2z_3,0),&\alpha=23,\quad 0<s<1/2,\\
  (2z_3,0),&\alpha=23,\quad 1/2<s<1.
 \end{cases}
\]
Substituting into the entropy integral and using $\ell(-a)=\ell(a)$ gives
\[
 J(q):=-\mathcal S[m_{\rm ref}]
 =\int_0^1\ell(|m_{\rm ref}'(s)|)\,\mathrm ds
 =\begin{cases}
  \ell(z_2),&\alpha=2,\\
  \ell(2z_3),&\alpha=23.
 \end{cases}
\]
These entropy costs are finite because $|z_2|<1$ for $\alpha=2$ and $2|z_3|<1$ for $\alpha=23$.

For any finite-entropy path $m$ in the closed window $|Q_\alpha(m)-q|_\infty\le\varepsilon$, linearity of $Q_\alpha$ and $Q_\alpha(m_{\rm ref})=q$ imply that $(1-\delta)m+\delta m_{\rm ref}$ lies in the open window $|Q_\alpha(\cdot)-q|_\infty<\varepsilon$ for every $\delta\in(0,1]$. Convexity of $-\mathcal S$ and continuity of $\mathcal E$ ensure that its free energy approaches that of $m$ as $\delta\downarrow0$. The two window infima therefore agree. Letting $\eta\downarrow0$, then $\varepsilon\downarrow0$, compactness of $\mathcal A$ and lower semicontinuity of $\lambda\mathcal E-\mathcal S$ give
\begin{equation}
\begin{aligned}
 S_{\alpha,\varepsilon}(q)
 &:=\lim_{n\to\infty}\frac1n\log A_{\alpha,n,\varepsilon}(q;\lambda)
 =-\inf_{|Q_\alpha(m)-q|_\infty<\varepsilon}\{\lambda\mathcal E[m]-\mathcal S[m]\},\\
 S_{\alpha,\varepsilon}(q)&\to
 S_\alpha(q):=-\inf_{Q_\alpha(m)=q}\{\lambda\mathcal E[m]-\mathcal S[m]\}
 \quad(\varepsilon\downarrow0).
\end{aligned}
\label{supp:fixed-window-limit}
\end{equation}
All infima are over $\mathcal A$; the fixed volume factor $(2\varepsilon)^{d_\alpha}$ does not affect the exponential rate.

\paragraph{Continuous densities.}
Under $\nu_n=\prod_i\mathrm d\theta_i/(2\pi)$, set $H_n=-U_{n,\lambda}$ and $d=d_\alpha$. Write
\begin{equation}
\begin{aligned}
 T_n=\frac1n\sum_{i=1}^nw_i\cos\theta_i=\hat Q_\alpha,\qquad
 G=\frac1n\sum_{i=1}^n\sin^2\theta_iw_iw_i^T,\\
 w_i=1\ (d=1),\qquad
 w_i=(1,c_i)^T,\quad c_i=\frac{2i-n-1}{n+1}\ (d=2).
\end{aligned}
\label{supp:density-notation}
\end{equation}
The matrix $G$ measures the sensitivity of the observables to bond rotations. For a continuous periodic weight $\Omega$, write
\[
 f_\Omega(q)=\int \Omega(\theta)\,\delta(T_n(\theta)-q)\,\mathrm d\nu_n(\theta),
 \qquad f_n=f_{e^{H_n}}.
\]
In ordinary integral terms, $\int_D f_\Omega(q)\,\mathrm dq$ is the total weight of configurations with $T_n(\theta)\in D$, for any box $D$ in observable space. We write $\mathbb P(B)=\int\mathbf1_B\,\mathrm d\nu_n$ for the probability of an angular event $B$.

\begin{steadylemma}[Smoothing by a few bonds]
\label{supp:density-smoothing-lemma}
For all sufficiently large $n$, $q\mapsto f_\Omega(q)$ is continuous for every continuous periodic angular weight $\Omega(\theta)$. Moreover, there is a choice of $J_n$ containing three indices for $d=1$, or six for $d=2$, such that, with $\theta_{\rm out}:=(\theta_i)_{i\notin J_n}$,
\[
 0\le\Omega(\theta)\le K\mathbf1_B(\theta_{\rm out})
 \quad\Longrightarrow\quad
 0\le f_\Omega(q)\le KCn^d\mathbb P(B)
 \quad(q\in\mathbb R^d).
\]
Here $K>0$, the event $B$ depends only on $\theta_{\rm out}$, and $C$ is independent of $n,q,\Omega,K,B$.
\end{steadylemma}

\begin{proof}
Choose $J_n$ as any three indices for $d=1$, or the six indices nearest $nj/7$, $j=1,\ldots,6$, for $d=2$. With $\theta_{\rm out}$ fixed, write
\[
 q_{\rm out}:=\frac1n\sum_{i\notin J_n}w_i\cos\theta_i,
 \qquad T_n=q_{\rm out}+\frac1n\sum_{i\in J_n}w_i\cos\theta_i.
\]
Let only $\theta_i$, $i\in J_n$, be random, independently and uniformly distributed on $[0,2\pi)$. Then $T_n$ has probability density
\[
 g_{J_n}(q\mid\theta_{\rm out})
 :=\int_{[0,2\pi)^{|J_n|}}
 \delta\!\left(q_{\rm out}+\frac1n\sum_{i\in J_n}w_i\cos\theta_i-q\right)
 \prod_{i\in J_n}\frac{\mathrm d\theta_i}{2\pi}.
\]
A single projection $X_i=\cos\theta_i$ has the arcsine law and its Fourier transform $\varphi$ is the Bessel function of the first kind of order zero. The Bessel asymptotics give $|\varphi(t)|\le C(1+|t|)^{-1/2}$; the same decay holds with any polynomial amplitude $P(\cos\theta)$, by expansion in $\cos(k\theta)$ and the fixed-order Bessel asymptotics~\cite{abramowitz1964handbook}.
 
Let $Y:=\sum_{i\in J_n}w_iX_i$. Independence gives its Fourier transform explicitly as
\[
 \mathbb E[e^{\mathrm is\cdot Y}]
 =\prod_{i\in J_n}\varphi(s\cdot w_i).
\]
For $d=1$, this is $\varphi(s)^3$. For $d=2$, the six selected coefficients satisfy $|c_i-c_j|\ge\delta_c>0$ for $i\ne j$, with $\delta_c$ independent of $n$ for large $n$. At most one index can satisfy $|s_1+c_is_2|<b|s|$ for a sufficiently small fixed $b>0$. Indeed, if two indices $i\ne j$ satisfied this inequality, subtracting the two expressions and using $|c_i|\le1$ would imply
\[
\begin{aligned}
 |s_2|&<\frac{2b}{\delta_c}|s|,\quad
 |s_1|&\le |s_1+c_is_2|+|c_i||s_2|
 <\left(1+\frac2{\delta_c}\right)b|s|,
\end{aligned}
\]
which contradicts $|s|^2=s_1^2+s_2^2$. Thus at least five factors satisfy
\[
 |\varphi(s_1+c_is_2)|\le C(1+|s|)^{-1/2},
\]
while any remaining factor is bounded by one. Consequently,
\[
 \left|\mathbb E[e^{\mathrm is\cdot Y}]\right|
 \le C\begin{cases}
 (1+|s|)^{-3/2},&d=1,\\
 (1+|s|)^{-5/2},&d=2.
 \end{cases}
\]
The first bound is integrable on $\mathbb R$; the second is integrable on $\mathbb R^2$.
Fourier inversion therefore gives $Y$ a continuous bounded density. Rescaling by $1/n$ and shifting by $q_{\rm out}$ shows that $g_{J_n}$ is continuous and
\[
 0\le g_{J_n}(q\mid\theta_{\rm out})\le Cn^d
 \qquad\text{uniformly in }q,\theta_{\rm out}.
\]

For $X=(X_i)\in[-1,1]^n$, define the sign-averaged weight
\[
 \overline\Omega(X)
 :=2^{-n}\!\sum_{\varepsilon\in\{-1,1\}^n}
 \Omega(\varepsilon_1\arccos X_1,\ldots,
        \varepsilon_n\arccos X_n),
\]
where the angles are interpreted periodically. This is continuous on $[-1,1]^n$. Since the changes of variables $\theta_i\mapsto-\theta_i$ preserve the uniform angular average and $T_n$, for every box $D$, we have
\[
 \int_{\{T_n\in D\}}\Omega(\theta)\,\mathrm d\nu_n
 =\int_{\{T_n\in D\}}\overline\Omega(X(\theta))\,\mathrm d\nu_n.
\]

Choose polynomials $P_m$ converging uniformly to $\overline\Omega$ on $[-1,1]^n$, and set $\Omega_m(\theta):=P_m(X(\theta))$. Each $P_m$ is a finite linear combination of products $\prod_i p_i(X_i)$, so it suffices to prove continuity for one such product. For each one-variable polynomial $p$, write $\varphi_p(t):=\frac1{2\pi}\int_0^{2\pi}
 p(\cos\theta)e^{\mathrm it\cos\theta}\,\mathrm d\theta$.
The polynomial-amplitude estimate above gives $|\varphi_p(t)|\le C_p(1+|t|)^{-1/2}$. For a product $P(X)=\prod_i p_i(X_i)$, set $\Omega_P(\theta):=P(X(\theta))$ and define its weighted density with $\theta_{\rm out}$ fixed by
\[
 h_P(q\mid\theta_{\rm out})
 :=\int_{[0,2\pi)^{|J_n|}}\Omega_P(\theta)
       \delta(T_n(\theta)-q)
       \prod_{i\in J_n}\frac{\mathrm d\theta_i}{2\pi}.
\]
Its Fourier transform in $q$ can be expressed as
\[
\begin{aligned}
 \widehat h_P(s\mid\theta_{\rm out})
 &:=\int_{\mathbb R^d}e^{\mathrm is\cdot q}
                 h_P(q\mid\theta_{\rm out})\,\mathrm dq =e^{\mathrm is\cdot q_{\rm out}}
   \left[\prod_{i\notin J_n}p_i(\cos\theta_i)\right]
   \prod_{i\in J_n}\varphi_{p_i}\!\left(\frac{s\cdot w_i}{n}\right).
\end{aligned}
\] 
The bracketed product is uniformly bounded because every $p_i$ is bounded on $[-1,1]$. Using the same three or six reserved factors as above therefore yields
\[
 |\widehat h_P(s\mid\theta_{\rm out})|
 \le M_{P,n}(s):=C_{P,n}
 \begin{cases}
 (1+|s|)^{-3/2},&d=1,\\
 (1+|s|)^{-5/2},&d=2.
 \end{cases}
\]
$M_{P,n}$ is integrable, and its Fourier inversion implies $h_P$ continuous in $q$, and the continuity estimate
\[
\begin{aligned}
 &|h_P(q+u\mid\theta_{\rm out})-h_P(q\mid\theta_{\rm out})| \le\frac1{(2\pi)^d}\int_{\mathbb R^d}
       |e^{-\mathrm is\cdot u}-1|M_{P,n}(s)\,\mathrm ds
       \longrightarrow0\qquad(u\to0).
\end{aligned}
\]
Since the right-hand side is independent of $\theta_{\rm out}$, it also bounds the change in the full density
\[
 f_{\Omega_P}(q)=\int h_P(q\mid\theta_{\rm out})
                \prod_{i\notin J_n}\frac{\mathrm d\theta_i}{2\pi}.
\]
Hence $f_{\Omega_P}$ is continuous. By linearity in the weight, the same holds for each $f_{\Omega_m}$.

The bound on $g_{J_n}$ also gives the unweighted density bound
\[
 f_1(q)=\int g_{J_n}(q\mid\theta_{\rm out})
             \prod_{i\notin J_n}\frac{\mathrm d\theta_i}{2\pi}
 \le Cn^d.
\]
For any box $D$, the defining integral identities for $f_{\Omega_m}$ and $f_{\Omega_k}$ give
\[
\begin{aligned}
 \left|\int_D(f_{\Omega_m}-f_{\Omega_k})(p)\,\mathrm dp\right|
 &=\left|\int_{\{T_n\in D\}}
       [P_m(X(\theta))-P_k(X(\theta))]\,\mathrm d\nu_n\right| \le\|P_m-P_k\|_\infty\int_D f_1(p)\,\mathrm dp.
\end{aligned}
\]
Take $D=D_r(q):=q+(-r,r)^d$, divide by its volume $(2r)^d$, and let $r\downarrow0$. Since $f_{\Omega_m}$, $f_{\Omega_k}$, and $f_1$ are continuous, their averages over $D_r(q)$ converge to their values at $q$. Hence
\[
\begin{aligned}
 |f_{\Omega_m}(q)-f_{\Omega_k}(q)|
 &\le\|P_m-P_k\|_\infty f_1(q)
 \le Cn^d\|P_m-P_k\|_\infty.
\end{aligned}
\]
At fixed $n$, uniform convergence of $P_m$ makes the right-hand side tend to zero as $m,k\to\infty$, uniformly in $q$. Thus the densities have a uniform limit, which we denote by $F(q):=\lim_{m\to\infty}f_{\Omega_m}(q)$.
As a uniform limit of continuous functions, $F$ is continuous. 
For each $m$ and each fixed bounded box $D$,
\[
 \int_D f_{\Omega_m}(q)\,\mathrm dq
 =\int_{\{T_n\in D\}}P_m(X(\theta))\,\mathrm d\nu_n.
\]
Both sides have limits as $m\to\infty$. Passing to the limit in the preceding identity, and then using the sign-averaging identity, we obtain 
\[
\begin{aligned}
 \int_D F(q)\,\mathrm dq
 &=\int_{\{T_n\in D\}}\overline\Omega(X(\theta))\,\mathrm d\nu_n =\int_{\{T_n\in D\}}\Omega(\theta)\,\mathrm d\nu_n.
\end{aligned}
\]
Thus $F$ is the continuous density $f_\Omega$.

Finally, assume $0\le\Omega(\theta)\le K\mathbf1_B(\theta_{\rm out})$.Averaging first over the angles in $J_n$ gives, we obtain for any box $D$,
\[
\begin{aligned}
 \int_D f_\Omega(q)\,\mathrm dq
 &\le K\int\mathbf1_B(\theta_{\rm out})
       \left[\int_D g_{J_n}(q\mid\theta_{\rm out})\,\mathrm dq\right]
       \prod_{i\notin J_n}\frac{\mathrm d\theta_i}{2\pi} \le KCn^d\mathbb P(B)|D|.
\end{aligned}
\]
Taking $D=D_r(q)$, dividing by $(2r)^d$, and letting $r\downarrow0$ concludes that
\[
 0\le f_\Omega(q)\le KCn^d\mathbb P(B).
\]
\end{proof}

\paragraph{Negligible degenerate configurations.}
We first show that configurations with small $\det G$ contribute an arbitrarily small exponential remainder to the density. Choose a smooth function $\chi:[0,\infty)\to[0,1]$ that is zero on $[0,1]$ and one on $[2,\infty)$, and set
\[
 \chi_\kappa(\theta)=\chi\!\left(\frac{\det G(\theta)}{\kappa}\right),
 \qquad b_{n,\kappa}=f_{(1-\chi_\kappa)e^{H_n}}.
\]

\begin{steadylemma}[An arbitrarily small exponential remainder]
\label{supp:degenerate-density-lemma}
For every $M>0$, there exist $\kappa\in(0,1/2]$ and $C_M>0$ such that, for all sufficiently large $n$ and all $q$,
\begin{equation}
 0\le b_{n,\kappa}(q)\le B_n:=C_Mn^de^{-Mn}.
 \label{supp:degenerate-density-bound}
\end{equation}
\end{steadylemma}

\begin{proof}
We need an upper bound on $H_n$ to control the Boltzmann weight. We also record its angular derivative bound for the transport argument below. The centered bead positions satisfy
\[
 r_k=\sum_{i=1}^n b_{ki}(\cos\theta_i,\sin\theta_i),
 \qquad b_{ki}=\mathbf1_{\{i\le k\}}-\frac{n+1-i}{n+1},
 \qquad k=0,\ldots,n.
\]
Since $|b_{ki}|\le1$, we have $|r_k|\le n$ and $|\partial_{\theta_i}r_k|\le1$. Thus
\[
\begin{aligned}
 H_n&=\frac{\lambda}{n^2}\sum_{k=0}^n(r_{k,x}^2-r_{k,y}^2)
 \le\lambda(n+1),\quad
 |\partial_{\theta_i}H_n|
 \le\frac{2\lambda}{n^2}\sum_{k=0}^n
       |r_k|\,|\partial_{\theta_i}r_k|
 \le\frac{2\lambda(n+1)}n.
\end{aligned}
\]
We then have 
\begin{equation}
 H_n\le2\lambda n,\qquad |\partial_{\theta_i}H_n|\le4\lambda.
 \label{supp:angular-energy-bounds}
\end{equation}

We first construct an event depending only on $\theta_{\rm out}$ so that Lemma~\ref{supp:density-smoothing-lemma} applies to the weight $(1-\chi_\kappa)e^{H_n}$. Define
\[
 G_{\rm out}=\frac1n\sum_{i\notin J_n}\sin^2\theta_iw_iw_i^T.
\]
Here matrix norms are operator norms. The definitions give
\begin{equation}
 0\le G_{\rm out}\le G,\qquad
 \|G_{\rm out}\|\le\operatorname{tr}G\le2.
 \label{supp:outside-matrix-bounds}
\end{equation}
For any $\kappa\in(0,1/2]$, set $\eta=\sqrt{2\kappa}$. Then $\det G\le2\kappa$ yields $\lambda_{\min}(G)\le\eta$. Since $G_{\rm out}\le G$ by Eq.~\eqref{supp:outside-matrix-bounds}, we have
\begin{equation}
 \{\det G\le2\kappa\}\subset\mathcal B_\eta,
 \qquad
 \mathcal B_\eta:=\{\lambda_{\min}(G_{\rm out})\le\eta\}.
 \label{supp:degenerate-event-inclusion}
\end{equation}
Since $1-\chi_\kappa$ vanishes for $\det G\ge2\kappa$, combining estimates~\eqref{supp:degenerate-event-inclusion} and~\eqref{supp:angular-energy-bounds} we obtain
\[
 0\le(1-\chi_\kappa)e^{H_n}
 \le e^{2\lambda n}\mathbf1_{\mathcal B_\eta}(\theta_{\rm out}).
\]
The event $\mathcal B_\eta$ depends only on $\theta_{\rm out}$. 
Using Lemma~\ref{supp:density-smoothing-lemma}, with taking $K=e^{2\lambda n}$ and $B=\mathcal B_\eta$, we obtain
\begin{equation}
 b_{n,\kappa}(q)\le Cn^de^{2\lambda n}\mathbb P(\mathcal B_\eta).
 \label{supp:degenerate-density-reduction}
\end{equation}
It remains to make this probability decay fast enough to compensate for $e^{2\lambda n}$.

To estimate $\mathbb P(\mathcal B_\eta)$, we will show that a small eigenvalue forces the sine squares of many independent angles to have a small sum. We first ensure that every direction has enough indices with a fixed lower bound on $|a\cdot w_i|^2$. For this purpose, define the deterministic matrix
\[
 K_{\rm out}=\frac1n\sum_{i\notin J_n}w_iw_i^T.
\]
Direct summation of the coefficients $w_i$ gives
\[
 \frac1n\sum_{i=1}^nw_iw_i^T
 =\begin{cases}
  1,&d=1,\\
  \operatorname{diag}\!\left(1,\dfrac{n-1}{3(n+1)}\right),&d=2.
 \end{cases}
\]
Deleting $|J_n|\le6$ terms changes this matrix by at most $12/n$ since $|w_i|^2\le2$. Therefore, for sufficiently large $n$, 
\begin{equation}
 K_{\rm out}\ge I_d/12.
 \label{supp:outside-gram-lower-bound}
\end{equation}
For any unit vector $a$, we set
\begin{equation}
 I_a=\{i\notin J_n:|a\cdot w_i|^2\ge1/24\}.
 \label{supp:direction-index-set}
\end{equation}
Using Eq.~\eqref{supp:outside-gram-lower-bound}, the definition in Eq.~\eqref{supp:direction-index-set}, and $|a\cdot w_i|^2\le2$, we obtain
\begin{equation}
 \frac n{12}\le\sum_{i\notin J_n}|a\cdot w_i|^2
 \le\frac n{24}+2|I_a|,
 \qquad |I_a|\ge\frac n{48}.
 \label{supp:direction-index-count}
\end{equation}

A unit eigenvector corresponding to $\lambda_{\min}(G_{\rm out})$ may depend on the angles. Using this vector as $a$ in Eq.~\eqref{supp:direction-index-set} could therefore make the index set depend on the angles. Instead, we first choose a finite set $\mathcal N_\eta$ of unit vectors, independently of the angles and of $n$, such that
\[
 \text{for every }v\in\mathbb R^d\text{ with }|v|=1,
 \quad \text{there is }a\in\mathcal N_\eta
 \text{ with }|a-v|\le\eta/4.
\]
For $d=1$, take $\mathcal N_\eta=\{-1,1\}$; for $d=2$, take sufficiently many equally spaced points on the unit circle. For each fixed $a\in\mathcal N_\eta$, the set $I_a$ in Eq.~\eqref{supp:direction-index-set} depends only on $a$ and the given coefficients $w_i$, and hence does not depend on the angles.

Now consider any angle configuration in $\mathcal B_\eta$. Choose a unit eigenvector $v$ for the smallest eigenvalue of $G_{\rm out}$, and then choose $a\in\mathcal N_\eta$ with $|a-v|\le\eta/4$. Using Eq.~\eqref{supp:outside-matrix-bounds}, we get
\[
 a^TG_{\rm out}a
 \le v^TG_{\rm out}v+2\|G_{\rm out}\|\,|a-v|
 \le\eta+4\frac\eta4=2\eta.
\]
For this $a$, the definition in Eq.~\eqref{supp:direction-index-set} implies
\[
 \frac1{24n}\sum_{i\in I_a}\sin^2\theta_i
 \le a^TG_{\rm out}a\le2\eta.
\]
Thus, for every configuration in $\mathcal B_\eta$, this inequality holds for at least one $a\in\mathcal N_\eta$. The choice of $a$ may vary with the configuration, so we take the union over all these fixed directions:
\begin{equation}
 \mathcal B_\eta
 \subset\bigcup_{a\in\mathcal N_\eta}
 \left\{\sum_{i\in I_a}\sin^2\theta_i\le48\eta n\right\}.
 \label{supp:degenerate-event-union}
\end{equation}
The bound $|I_a|\ge n/48$ from Eq.~\eqref{supp:direction-index-count} applies to every fixed index set in this union.

We next bound the probability of each scalar-sum event in Eq.~\eqref{supp:degenerate-event-union}. For $t>0$, define
\[
 \psi(t)=\frac1{2\pi}\int_0^{2\pi}e^{-t\sin^2\theta}\,\mathrm d\theta.
\]
Splitting this integral into $|\sin\theta|\le r$ and its complement gives
\[
 \psi(t)\le\frac2\pi\arcsin r+e^{-tr^2},
 \qquad 0<r<1.
\]
Letting $t\to\infty$ and then $r\downarrow0$ shows that $\psi(t)\to0$. The exponential Markov inequality and independence now give, for each fixed $a\in\mathcal N_\eta$,
\begin{equation}
\begin{aligned}
 \mathbb P\!\left\{\sum_{i\in I_a}\sin^2\theta_i\le48\eta n\right\}
 &\le e^{48t\eta n}
       \mathbb E\!\left[e^{-t\sum_{i\in I_a}\sin^2\theta_i}\right] =e^{48t\eta n}\psi(t)^{|I_a|} \le\exp\!\left\{n\left[48t\eta+\frac1{48}\log\psi(t)\right]\right\}.
\end{aligned}
 \label{supp:direction-tail-estimate}
\end{equation}
The last inequality uses $0<\psi(t)\le1$ and the cardinality bound in Eq.~\eqref{supp:direction-index-count}. Combining the event inclusion in Eq.~\eqref{supp:degenerate-event-union} with Eq.~\eqref{supp:direction-tail-estimate}, and summing over the finite set of directions, yields
\begin{equation}
 \mathbb P(\mathcal B_\eta)
 \le|\mathcal N_\eta|
 \exp\!\left\{n\left[48t\eta+\frac1{48}\log\psi(t)\right]\right\}.
 \label{supp:degenerate-probability-bound}
\end{equation}
To compensate for the factor $e^{2\lambda n}$ in Eq.~\eqref{supp:degenerate-density-reduction} and obtain the prescribed rate $M$, first choose $t$ so that $(\log\psi(t))/48<-(M+2\lambda+2)$, and then choose $\kappa$ so that $48t\sqrt{2\kappa}<1$. With these choices, Eq.~\eqref{supp:degenerate-probability-bound} gives
\begin{equation}
 \mathbb P(\mathcal B_\eta)
 \le|\mathcal N_\eta|e^{-n(M+2\lambda+1)},
 \label{supp:degenerate-probability-rate}
\end{equation}
and both $\kappa$ and $|\mathcal N_\eta|$ are independent of $n$.

Substituting Eq.~\eqref{supp:degenerate-probability-rate} into the density bound in Eq.~\eqref{supp:degenerate-density-reduction}, we conclude that
\[
\begin{aligned}
 b_{n,\kappa}(q)
 &\le Cn^de^{2\lambda n}|\mathcal N_\eta|e^{-n(M+2\lambda+1)} =C|\mathcal N_\eta|n^de^{-n(M+1)}
 \le C_Mn^de^{-Mn},
\end{aligned}
\]
where the fixed factor $C|\mathcal N_\eta|$ is absorbed into $C_M$.
\end{proof}

\paragraph{Density comparison by bond rotations.}
For configurations with $\chi_\kappa>0$, we use bond rotations to translate $T_n$ from $q$ to a nearby $p$. We bound the changes in the Boltzmann weight and angular volume to compare the densities at these points. Lemma~\ref{supp:degenerate-density-lemma} bounds the remaining contribution to the density.

\begin{steadylemma}[Density comparison by transport]
\label{supp:transport-lemma}
Let $\kappa$ and $B_n$ be as in Lemma~\ref{supp:degenerate-density-lemma}, and set $\delta_\kappa=\kappa/(16\sqrt2)$ and $L_\kappa=8\lambda/\sqrt\kappa+8\sqrt2/\kappa$. For all sufficiently large $n$ and all $p,q\in\mathbb R^d$ with $|p-q|<\delta_\kappa$, we have
\begin{equation}
 f_n(q)\le e^{nL_\kappa|p-q|}f_n(p)+B_n.
 \label{supp:transport-comparison}
\end{equation}
\end{steadylemma}

\begin{proof}
For a unit vector $a\in\mathbb R^d$ and $\det G>0$, define
\[
 b_i=w_i^TG^{-1}a,\qquad d_i=w_i^TG^{-1}w_i,
 \qquad\mathcal V_i=-\sin\theta_i b_i.
\]
We write $\mathcal VA=\sum_i\mathcal V_i\partial_{\theta_i}A$ for the derivative of an angular function $A$ along the vector field $\mathcal V$. Since $\partial_{\theta_i}T_n=-n^{-1}\sin\theta_iw_i$,
\[
 \mathcal VT_n
 =\frac1n\sum_i\sin^2\theta_iw_iw_i^TG^{-1}a
 =GG^{-1}a=a.
\]
Thus moving the angles with velocity $\mathcal V$ translates $T_n$ at unit speed in direction $a$.

To bound the change in $\det G$ and the divergence of $\mathcal V$, differentiate $G$ and its inverse:
\[
 \partial_{\theta_i}G=\frac2n\sin\theta_i\cos\theta_iw_iw_i^T,
 \qquad
 \partial_{\theta_i}G^{-1}
 =-G^{-1}(\partial_{\theta_i}G)G^{-1}.
\]
Hence $\partial_{\theta_i}\log\det G=(2/n)\sin\theta_i\cos\theta_i d_i$ and $\partial_{\theta_i}b_i=-(2/n)\sin\theta_i\cos\theta_i d_ib_i$. Substituting into $\mathcal V\log\det G$ and $\operatorname{div}\mathcal V:=\sum_i\partial_{\theta_i}\mathcal V_i$ yields
\begin{equation}
\begin{aligned}
 \mathcal V\log\det G
 &=-\frac2n\sum_i\sin^2\theta_i\cos\theta_i b_id_i,\\
 \operatorname{div}\mathcal V
 &=-\sum_i\cos\theta_i b_i
   +\frac2n\sum_i\sin^2\theta_i\cos\theta_i d_ib_i,\\
 \frac1n\sum_i\sin^2\theta_i d_i
 &=\operatorname{tr}(G^{-1}G)=d.
\end{aligned}
\label{supp:transport-derivatives}
\end{equation}
Set $\mu=\lambda_{\min}(G)$. The first and third identities in Eq.~\eqref{supp:transport-derivatives}, together with $d_i\ge0$ and $|b_i|\le\sqrt2/\mu$, give
\begin{equation}
 |\mathcal V\det G|
 \le\frac{2\det G}{n}\max_i|b_i|
          \sum_i\sin^2\theta_i d_i
 \le2d\sqrt2\,\frac{\det G}{\mu}\le8\sqrt2.
 \label{supp:transport-determinant-speed}
\end{equation}
Here $\det G/\mu=1$ for $d=1$, while for $d=2$ it equals $\lambda_{\max}(G)\le\operatorname{tr}G\le2$.

To obtain a smooth invertible flow for the angular change of variables, we extend $\mathcal V$ to all angles. Choose a smooth $\rho:[0,\infty)\to[0,1]$ that is zero on $[0,1/4]$ and one on $[1/2,\infty)$. Define
\[
 \widetilde{\mathcal V}=\rho(\det G/\kappa)\mathcal V
 \quad(\det G>0),\qquad
 \widetilde{\mathcal V}=0\quad(\det G=0),
\]
and let $\Phi_t(\theta)$ solve
\[
 \frac{\mathrm d}{\mathrm dt}\Phi_t(\theta)
 =\widetilde{\mathcal V}(\Phi_t(\theta)),
 \qquad \Phi_0(\theta)=\theta.
\]
All angles are interpreted periodically. The extended field is smooth and bounded, so its flow exists for all $t$ and has inverse $\Phi_{-t}$.

Consider a trajectory starting at $\theta$ with $\chi_\kappa(\theta)>0$, so $\det G(\theta)>\kappa$. Along this trajectory, the cutoff $\rho(\det G/\kappa)$ equals one until $\det G$ first reaches $\kappa/2$. Equation~\eqref{supp:transport-determinant-speed} then gives
\[
 \det G(\Phi_t\theta)
 \ge\det G(\theta)-8\sqrt2\,t>\frac\kappa2,
 \qquad 0\le t<\delta_\kappa=\frac\kappa{16\sqrt2}.
\]
Thus $\det G$ remains above $\kappa/2$ for $0\le t<\delta_\kappa$. On this interval,
\[
 T_n(\Phi_t\theta)=T_n(\theta)+ta,\qquad
 \mu\ge\frac{\det G}{2}>\frac\kappa4.
\]
The eigenvalue bound uses $\det G/\mu\le2$.

To bound the change in $H_n$ along these trajectories, substitute the angular velocities and use $\sum_i\sin^2\theta_iw_iw_i^T=nG$:
\[
\begin{aligned}
 \sum_i\mathcal V_i^2
 &=\sum_i\sin^2\theta_i(w_i^TG^{-1}a)^2 =a^TG^{-1}\!\left(\sum_i\sin^2\theta_iw_iw_i^T\right)G^{-1}a =na^TG^{-1}GG^{-1}a
 =na^TG^{-1}a\le\frac n\mu\le\frac{4n}{\kappa}.
\end{aligned}
\]
Cauchy--Schwarz yields $\sum_i|\mathcal V_i|\le2n/\sqrt\kappa$, so Eq.~\eqref{supp:angular-energy-bounds} implies
\[
 |\mathcal VH_n|\le4\lambda\sum_i|\mathcal V_i|
 \le\frac{8\lambda n}{\sqrt\kappa}.
\]
For the angular volume change, we bound $\operatorname{div}\mathcal V$. Using $|b_i|\le4\sqrt2/\kappa$ in Eq.~\eqref{supp:transport-derivatives} we obtain
\[
\begin{aligned}
 |\operatorname{div}\mathcal V|
 &\le n\max_i|b_i|
   +\frac2n\max_i|b_i|\sum_i\sin^2\theta_i d_i \le\frac{4\sqrt2 n+16\sqrt2}{\kappa}
 \le\frac{8\sqrt2 n}{\kappa}\qquad(n\ge4).
\end{aligned}
\]

The Jacobian of the angular change of variables is
\[
 j_t(\theta)=|\det D_\theta\Phi_t(\theta)|.
\]
Differentiating the flow equation with respect to the initial angles gives
\[
 \frac{\mathrm d}{\mathrm dt}D_\theta\Phi_t(\theta)
 =D\widetilde{\mathcal V}(\Phi_t\theta)\,D_\theta\Phi_t(\theta),
\]
where $D\widetilde{\mathcal V}$ is the matrix with entries $\partial_{\theta_j}\widetilde{\mathcal V}_i$. The matrix $D_\theta\Phi_t$ is invertible because $\Phi_t$ has the smooth inverse $\Phi_{-t}$. Jacobi's formula and cyclic invariance of the trace give
\[
\begin{aligned}
 \frac{\mathrm d}{\mathrm dt}\log j_t(\theta)
 &=\operatorname{tr}\!\left[(D_\theta\Phi_t)^{-1}
          \frac{\mathrm d}{\mathrm dt}D_\theta\Phi_t\right] =\operatorname{tr}\!\left[(D_\theta\Phi_t)^{-1}
       D\widetilde{\mathcal V}(\Phi_t\theta)D_\theta\Phi_t\right] =\operatorname{tr}D\widetilde{\mathcal V}(\Phi_t\theta)
 =(\operatorname{div}\widetilde{\mathcal V})(\Phi_t\theta).
\end{aligned}
\]
For trajectories starting at $\theta$ with $\chi_\kappa(\theta)>0$, the cutoff $\rho(\det G/\kappa)$ equals one in a neighborhood of every trajectory point for $0\le t<\delta_\kappa$. Hence $\widetilde{\mathcal V}$ and its derivatives agree with $\mathcal V$ there. Combining the bounds on $\mathcal VH_n$ and $\operatorname{div}\mathcal V$ gives
\[
 \left|\frac{\mathrm d}{\mathrm dt}
 \log\!\left[e^{H_n(\Phi_t\theta)}j_t(\theta)\right]\right|
 =|(\mathcal VH_n+\operatorname{div}\mathcal V)(\Phi_t\theta)|
 \le nL_\kappa.
\]
Integrating from $0$ to $t$, with $j_0=1$, yields
\begin{equation}
 e^{H_n(\theta)}
 \le e^{nL_\kappa t}e^{H_n(\Phi_t\theta)}j_t(\theta),
 \qquad \chi_\kappa(\theta)>0,\quad 0\le t<\delta_\kappa.
 \label{supp:transport-weight-volume-bound}
\end{equation}
For any box $D\subset\mathbb R^d$, the identity $T_n(\Phi_t\theta)=T_n(\theta)+ta$ shows that $\Phi_t$ maps $\{T_n\in D,\chi_\kappa>0\}$ into $\{T_n\in D+ta\}$. Integrating Eq.~\eqref{supp:transport-weight-volume-bound}, using $0\le\chi_\kappa\le1$, and changing angular variables by $\Phi_t$ gives
\[
\begin{aligned}
 \int_D f_{\chi_\kappa e^{H_n}}(x)\,\mathrm dx
 &=\int_{\{T_n\in D\}}\chi_\kappa e^{H_n}\,\mathrm d\nu_n \le e^{nL_\kappa t}
 \int_{\{T_n\in D,\chi_\kappa>0\}}
 e^{H_n(\Phi_t\theta)}j_t(\theta)\,\mathrm d\nu_n(\theta) \le e^{nL_\kappa t}\int_{D+ta}f_n(x)\,\mathrm dx.
\end{aligned}
\]
Take $D=D_r(q)=q+(-r,r)^d$, divide by $(2r)^d$, and let $r\downarrow0$. The densities $f_{\chi_\kappa e^{H_n}}$ and $f_n$ are continuous by Lemma~\ref{supp:density-smoothing-lemma}, so
\[
 f_n(q)-b_{n,\kappa}(q)
 =f_{\chi_\kappa e^{H_n}}(q)
 \le e^{nL_\kappa t}f_n(q+ta).
\]
For $p\ne q$, choose $a=(p-q)/|p-q|$ and $t=|p-q|$. The bound $b_{n,\kappa}(q)\le B_n$ from Lemma~\ref{supp:degenerate-density-lemma} then concludes Eq.~\eqref{supp:transport-comparison}.
\end{proof}

\paragraph{From window averages to pointwise free energies.}
\begin{proof}[Proof of Theorem~\ref{supp:limit-interchange}]
Fix $q$ as in Theorem~\ref{supp:limit-interchange} and abbreviate $A_{n,\varepsilon}=A_{\alpha,n,\varepsilon}(q;\lambda)$.
For $m_{\rm ref}$ and $J(q)$ defined above, we have $Q_\alpha(m_{\rm ref})=q$, $\mathcal E[m_{\rm ref}]\le0$, and $J(q)<\infty$. Evaluating the variational problems at $m_{\rm ref}$ gives
\[
 S_{\alpha,\varepsilon}(q)\ge S_\alpha(q)\ge-J(q)>-\infty.
\]
Choose $M>J(q)+5$ and the corresponding $\kappa,B_n$ from Lemma~\ref{supp:degenerate-density-lemma}. These choices are independent of $\varepsilon$.

For $|p-q|_\infty<\varepsilon<\delta_\kappa/\sqrt d$, we have $|p-q|\le\sqrt d\,\varepsilon$. Applying Lemma~\ref{supp:transport-lemma} in both directions and using $f_n\ge0$ gives
\[
 f_n(q)\le e^{nC_\kappa\varepsilon}f_n(p)+B_n,
 \qquad
 f_n(p)\le e^{nC_\kappa\varepsilon}f_n(q)+B_n,
 \qquad C_\kappa=\sqrt d\,L_\kappa.
\]
Averaging over $|p-q|_\infty<\varepsilon$ and rearranging yields
\begin{equation}
 e^{-nC_\kappa\varepsilon}[f_n(q)-B_n]
 \le A_{n,\varepsilon}
 \le e^{nC_\kappa\varepsilon}f_n(q)+B_n,
 \qquad 0<\varepsilon<\delta_\kappa/\sqrt d.
 \label{supp:averaged-transport}
\end{equation}

To show that $B_n/f_n(q)\to0$ independently of $\varepsilon$, we first bound $f_n(q)$ from below. Fix $0<\varepsilon_0<\delta_\kappa/\sqrt d$ with $C_\kappa\varepsilon_0\le1$. Equation~\eqref{supp:fixed-window-limit} and $S_{\alpha,\varepsilon_0}(q)\ge-J(q)$ give
\[
 A_{n,\varepsilon_0}\ge e^{-n[J(q)+1]}
 \qquad\text{for sufficiently large }n.
\]
Since $M>J(q)+5$, we also have $B_n\le\frac12e^{-n[J(q)+1]}$ for sufficiently large $n$. The upper bound in Eq.~\eqref{supp:averaged-transport} at $\varepsilon=\varepsilon_0$ therefore implies
\begin{equation}
\begin{aligned}
 f_n(q)
 &\ge e^{-nC_\kappa\varepsilon_0}[A_{n,\varepsilon_0}-B_n] \ge\tfrac12e^{-n[J(q)+2]}>0.
\end{aligned}
 \label{supp:positive-density-bound}
\end{equation}
Combining Eqs.~\eqref{supp:positive-density-bound} and~\eqref{supp:degenerate-density-bound} gives
\[
 \mathcal R_n:=\frac{B_n}{f_n(q)}
 \le2C_Mn^de^{-n[M-J(q)-2]}\longrightarrow0,
\]
This bound is independent of $\varepsilon$.
Dividing Eq.~\eqref{supp:averaged-transport} by $f_n(q)$ and using $e^{nC_\kappa\varepsilon}\ge1$ we obtain
\begin{equation}
 e^{-nC_\kappa\varepsilon}(1-\mathcal R_n)
 \le\frac{A_{n,\varepsilon}}{f_n(q)}
 \le e^{nC_\kappa\varepsilon}(1+\mathcal R_n).
 \label{supp:uniform-ratio-bound}
\end{equation}
Since $\mathcal R_n<1$ for sufficiently large $n$, the lower bound in Eq.~\eqref{supp:uniform-ratio-bound} is positive. Taking logarithms gives
\[
\begin{aligned}
 -C_\kappa\varepsilon+\frac1n\log(1-\mathcal R_n)
 &\le\frac1n\log A_{n,\varepsilon}-\frac1n\log f_n(q) \le C_\kappa\varepsilon+\frac1n\log(1+\mathcal R_n).
\end{aligned}
\]
Since $\log(1+\mathcal R_n)\le-\log(1-\mathcal R_n)$, this yields
\begin{equation}
 \left|\frac1n\log A_{n,\varepsilon}-\frac1n\log f_n(q)\right|
 \le C_\kappa\varepsilon+\xi_n,
 \qquad
 \xi_n:=-\frac1n\log(1-\mathcal R_n)\longrightarrow0.
 \label{supp:uniform-log-comparison}
\end{equation}
The error $\xi_n$ is independent of $\varepsilon$. For each fixed $0<\varepsilon<\delta_\kappa/\sqrt d$, combining estimates~\eqref{supp:uniform-log-comparison} and~\eqref{supp:fixed-window-limit}, we obtain
\[
\begin{aligned}
 S_{\alpha,\varepsilon}(q)-C_\kappa\varepsilon
 &\le\liminf_{n\to\infty}\frac1n\log f_n(q)\le\limsup_{n\to\infty}\frac1n\log f_n(q)
 \le S_{\alpha,\varepsilon}(q)+C_\kappa\varepsilon.
\end{aligned}
\]
Letting $\varepsilon\downarrow0$ and using Eq.~\eqref{supp:fixed-window-limit} shows that
\begin{equation}
 \lim_{n\to\infty}\frac1n\log f_n(q)=S_\alpha(q).
 \label{supp:pointwise-log-limit}
\end{equation}
To take $\varepsilon\downarrow0$ before $n\to\infty$, fix $n$ large enough for Eq.~\eqref{supp:positive-density-bound}. The continuity of $f_n$ from Lemma~\ref{supp:density-smoothing-lemma} and the bound $f_n(q)>0$ imply
\[
 \lim_{\varepsilon\downarrow0}\frac1n\log A_{n,\varepsilon}
 =\frac1n\log f_n(q).
\]
Taking $n\to\infty$ and using Eq.~\eqref{supp:pointwise-log-limit}, we obtain $\lim_{n\to\infty}\lim_{\varepsilon\downarrow0}n^{-1}\log A_{n,\varepsilon}=S_\alpha(q)$. Equation~\eqref{supp:fixed-window-limit} yields $\lim_{\varepsilon\downarrow0}\lim_{n\to\infty}n^{-1}\log A_{n,\varepsilon}=S_\alpha(q)$ for the opposite order of limits. Thus both iterated limits agree,which concludes Theorem~\ref{supp:limit-interchange}.
\end{proof}

\subsection{Proof of the minimizing-interval theorem}
\label{supp:steady-transitions}
\begin{proof}
We begin with the scalar variational problem in Eq.~\eqref{supp:free-energy-limits}. Write its constrained value as
\[
G_\lambda(c)
:=\inf_{\substack{u\in\mathcal A_1\\u(1)=0,\ -2\bar u=c}}
\mathcal F^{(1)}_\lambda[u],\qquad c\in\mathbb R.
\]
Equation~\eqref{supp:free-energy-limits} implies $G_\lambda(c)=F_{23,\infty}(0,c;\lambda)$ for $|c|<1/2$; we verify the boundary and exterior values below. First, drop the mean constraint and define
\begin{equation}
\mathcal M_\lambda
=\operatorname*{argmin}_{\substack{u\in\mathcal A_1\\u(1)=0}}
\mathcal F^{(1)}_\lambda[u].
\label{eq:supp:transition-periodic-functional}
\end{equation}
The admissible set is compact in the uniform topology by the Arzel\`a--Ascoli theorem. On this set, the convex entropy integral is lower semicontinuous, while the quadratic energy term is continuous and bounded. Since the entropy term is nonnegative and $\mathcal F^{(1)}_\lambda[0]=0$, the minimum $m_\lambda :=\min_{u\in\mathcal A_1,\,u(1)=0}\mathcal F^{(1)}_\lambda[u]\le0$ is finite, $\mathcal M_\lambda$ is nonempty and compact, and its elements have finite entropy cost. Fixing the mean restricts the admissible set, so $G_\lambda(c)\ge m_\lambda$ for every $c \in\mathbb R$. If $u\in\mathcal M_\lambda$ and $c=-2\bar u$, then
\[
m_\lambda\le G_\lambda(c)
\le\mathcal F^{(1)}_\lambda[u]=m_\lambda.
\]
Thus each such $c$ minimizes $G_\lambda$, with $\min_cG_\lambda(c)=m_\lambda$. Conversely, if $c$ minimizes $G_\lambda$, the fixed-mean  set is compact with finite infimum $G_\lambda(c)=m_\lambda$, which is attained by a path $u_c \in\mathcal M_\lambda$ satisfying $c=-2\bar u_c$. Therefore,
\[
\operatorname*{argmin}_{c\in\mathbb R}G_\lambda(c)
=\{-2\bar u:u\in\mathcal M_\lambda\}.
\]

For $u\in\mathcal M_\lambda$, let $\widetilde u$ be its $1$-periodic extension and shift the starting point by $a\in[0,1]$. Subtracting $u(a)$ preserves the zero endpoint values, and the shifted paths satisfy
\begin{equation}
\begin{aligned}
u_a(s)&=\widetilde u(s+a)-u(a),\qquad u_a(0)=u_a(1)=0,\qquad
\mathcal F^{(1)}_\lambda[u_a]&=\mathcal F^{(1)}_\lambda[u],\qquad
-2\bar u_a=2\bigl(u(a)-\bar u\bigr).
\end{aligned}
\label{eq:supp:transition-shift}
\end{equation}
Indeed, the derivative is shifted over one period, and subtracting a constant leaves the variance unchanged. Thus $u_a\in\mathcal M_\lambda$. As $a$ varies over $[0,1]$, continuity of $u$ ensures that $-2\bar u_a$ sweeps out $[2(\min u-\bar u),2(\max u-\bar u)]$, which contains $0$. Every such interval consists of minimizing values of $c$, and every minimizing value belongs to one of them by taking $a=0$ (since $u_0=u$). Their union over $u\in\mathcal M_\lambda$ is therefore an interval containing $0$. It is compact because the mean map $u\mapsto\bar u$ is continuous on $\mathcal M_\lambda$, and symmetric because $u\mapsto-u$ preserves the functional. Consequently,
\[
\operatorname*{argmin}_{c\in\mathbb R}G_\lambda(c)
=[-a_\lambda,a_\lambda],
\]
where
\begin{equation}
a_\lambda=2\max_{u\in\mathcal M_\lambda}|\bar u|.
\label{eq:supp:transition-half-width}
\end{equation}

We next determine when the interval opens. The series $I_0(h)=\sum_{k\ge0}(h^2/4)^k/(k!)^2$ gives $I_0(h)<e^{h^2/4}$ for $h\ne0$. Taking $h=2v$ in the Legendre supremum gives $\ell(v)>v^2$ for $v\ne0$. Set $K(h)=\log I_0(h)$. Since $K$ is strictly convex and $K''(0)=1/2$, the Legendre maximizer near $v=0$ solves $K'(h)=v$. The series gives
\[
\begin{aligned}
K(h)&=\frac{h^2}{4}-\frac{h^4}{64}+O(h^6),\qquad
K'(h)=\frac h2-\frac{h^3}{16}+O(h^5),\\
h(v)&=2v+v^3+O(v^5).
\end{aligned}
\]
Substitution into $\ell(v)=vh(v)-K(h(v))$ yields
\begin{equation}
\ell(0)=0,\qquad \ell(v)>v^2\quad(v\ne0),\qquad
\ell(v)=v^2+\tfrac14v^4+O(v^6)\quad(v\to0).
\label{eq:supp:transition-entropy-bound}
\end{equation}
For nonzero $u\in\mathcal A_1$ with $u(1)=0$, the periodic Wirtinger inequality gives $\int_0^1|u'|^2\,\mathrm ds\ge4\pi^2\int_0^1(u-\bar u)^2\,\mathrm ds$. Together with Eq.~\eqref{eq:supp:transition-entropy-bound}, this implies $\mathcal F^{(1)}_\lambda[u]>0$ when $\lambda\le4\pi^2$. Hence $\mathcal M_\lambda=\{0\}$ and $a_\lambda=0$ in this range. For $\lambda>4\pi^2$, the trial path $u_\varepsilon(s)=\varepsilon\sin(2\pi s)$ is admissible for $|\varepsilon|\le(2\pi)^{-1}$ and satisfies
\begin{equation}
\mathcal F^{(1)}_\lambda[u_\varepsilon]
=\frac{\varepsilon^2}{2}(4\pi^2-\lambda)+O(\varepsilon^4)<0
\label{eq:supp:transition-periodic-trial}
\end{equation}
for sufficiently small $\varepsilon\ne0$. Thus $m_\lambda<0$, and every $u\in\mathcal M_\lambda$ is nonconstant. Choosing $a \in[0,1]$ with $u(a)\ne\bar u$ in Eq.~\eqref{eq:supp:transition-shift} then gives a minimizing path $u_a\in\mathcal M_\lambda$ with nonzero mean. Consequently, $a_\lambda>0$.

Finally, we identify $G_\lambda$ with the density-defined free energy on the remaining values of $c$. For $h>0$ and $0<r<\pi$, splitting the angular integral gives $e^{-h}I_0(h)\le r/\pi+e^{-h(1-\cos r)}$. Letting $h\to\infty$ and then $r\downarrow0$ gives $e^{-h}I_0(h)\to0$, so the Legendre definition yields $\ell(\pm1)=+\infty$. Any finite-entropy $u\in\mathcal A_1$ with $u(1)=0$ then has $|u'|<1$ almost everywhere; integration from either endpoint gives $|u(s)|<\min(s,1-s)$ for $0<s<1$, and hence $|\bar u|<1/4$. Thus $G_\lambda(c)=+\infty$ for $|c|\ge1/2$. Since $|\hat z_3|\le n^{-1}\sum_i|c_i|\le n/[2(n+1)]<1/2$, the density-defined free energy is also $+\infty$ there. Together with Eq.~\eqref{supp:free-energy-limits} on $|c|<1/2$, this gives $G_\lambda(c)=F_{23,\infty}(0,c;\lambda)$ for all $c \in\mathbb R$, proving Eq.~\eqref{eq:supp:transition-exact-interval} and the stated threshold.
\end{proof}

\paragraph{First transition.}
For the endpoint observable, $|\hat z_2|\le1$, so its continuous weighted density vanishes for $|z_2|>1$ and, by continuity, at $z_2=\pm1$. In the scalar problem, $u(1)=\pm1$ forces $u'=\pm1$ almost everywhere and hence infinite entropy cost, while $|u(1)|>1$ is impossible. Thus Eq.~\eqref{supp:free-energy-limits} extends to
\[
F_{2,\infty}(z_2;\lambda)
=\inf_{\substack{u\in\mathcal A_1\\u(1)=z_2}}
\mathcal F^{(1)}_\lambda[u],\qquad z_2\in\mathbb R.
\]

A similar argument determines the first transition threshold $\lambda_1=\pi^2$. Since the admissible set $\mathcal A_1$ is compact in the uniform topology and $\mathcal F^{(1)}_\lambda$ is lower semicontinuous, $\mathcal F^{(1)}_\lambda$ attains its global minimum on $\mathcal A_1$, and $\operatorname*{argmin}_{z_2\in\mathbb R}F_{2,\infty}(z_2;\lambda)=\{u(1):u\in\operatorname*{argmin}_{\mathcal A_1}\mathcal F^{(1)}_\lambda\}$. For $w=u-\bar u$ without the periodic boundary condition, expanding in the cosine basis $\{\cos(\pi ks)\}_{k\ge1}$ on $[0,1]$ yields the sharp Poincar\'e inequality $\int_0^1|w'|^2\,\mathrm ds\ge\pi^2\int_0^1w^2\,\mathrm ds$. By Eq.~\eqref{eq:supp:transition-entropy-bound}, $u=0$ is therefore the unique global minimizer of $\mathcal F^{(1)}_\lambda$ for all $\lambda\le\pi^2$. For $\lambda>\pi^2$, the trial path $u_\varepsilon(s)=\varepsilon[\cos(\pi s)-1]$ satisfies $\mathcal F^{(1)}_\lambda[u_\varepsilon]=\frac12\varepsilon^2(\pi^2-\lambda)+O(\varepsilon^4)<0$ for sufficiently small $\varepsilon\ne0$, so any global minimizer is nonconstant. Moreover, every global minimizer $u$ must be monotone. Indeed, replacing $u$ with $v(t)=\int_0^t|u'(s)|\mathrm ds$ preserves the entropy cost while increasing the variance:
\begin{equation}
\int_0^1(u-\bar u)^2\mathrm dt
=\frac12\int_0^1\int_0^1[u(t)-u(s)]^2\mathrm ds\mathrm dt
\le\frac12\int_0^1\int_0^1[v(t)-v(s)]^2\mathrm ds\mathrm dt
=\int_0^1(v-\bar v)^2\mathrm dt.
\label{eq:supp:transition-pairwise-variance}
\end{equation}
If $u'$ takes both positive and negative values on sets of positive measure, then $|u(1)-u(0)|<v(1)-v(0)$. By continuity, $|u(t)-u(s)|<|v(t)-v(s)|$ holds on a neighborhood of $(s,t)=(0,1)$, making the inequality in Eq.~\eqref{eq:supp:transition-pairwise-variance} strict and thus strictly decreasing $\mathcal F^{(1)}_\lambda$. Consequently, for $\lambda>\pi^2$, every global minimizer is nonconstant and monotone, and thus has a nonzero endpoint value $u(1)$, with opposite signs occurring in pairs by the reflection symmetry $u\mapsto-u$.

\subsection{Numerical evaluation}
\label{supp:steady-numerics} 

The scalar variational problems in Eq.~\eqref{supp:free-energy-limits}
are evaluated by a piecewise-linear discretization on $M$ uniform
intervals of length $h=1/M$. Writing $v_j$ for the slope on interval
$j$, we set
\begin{equation}
 u_0=0,\qquad
 u_j=h\sum_{i=1}^j v_i,\qquad
 |v_j|\le 1-10^{-6},\qquad j=1,\ldots,M.
 \label{eq:supp-steady-path}
\end{equation}
With trapezoidal weights $\omega_0=\omega_M=1/2$ and
$\omega_j=1$ otherwise, the discrete functional is
\begin{equation}
\begin{aligned}
 \bar u_M=h\sum_{j=0}^M\omega_j u_j,\qquad
 \mathcal F^{(1)}_{\lambda,M}(v)
 =h\sum_{j=1}^M\ell(v_j)
 -\lambda h\sum_{j=0}^M\omega_j(u_j-\bar u_M)^2.
\end{aligned}
 \label{eq:supp-discrete-objective}
\end{equation}
Both problems impose $u_M=z_2$. The joint problem also
imposes the mean constraint
\begin{equation}
 h\sum_{j=1}^M(1-s_{j-1/2})v_j=\frac{z_2-z_3}{2},
 \qquad s_{j-1/2}=(j-\tfrac12)h,
 \label{eq:supp-steady-constraints}
\end{equation}
and the section used in the main text is obtained by setting $z_2=0$. 

For $0<|v|<1$, we evaluate the entropy density from
\begin{equation*}
 \ell(v)=h_v|v|-\log I_0(h_v),\qquad
 \frac{I_1(h_v)}{I_0(h_v)}=|v|,
\end{equation*}
with $h_v\ge0$ determined numerically and $\ell(0)=0$.
Here $I_0$ and $I_1$ are modified Bessel functions.
The constrained minimization is performed with SLSQP to a tolerance
of $10^{-9}$. The endpoint solutions are continued in $|z_2|$ from
constant-slope initializations. For the joint problem, continuation
is supplemented by multiple projected initializations; solutions with residuals
above $2\times10^{-7}$ are discarded, and the smallest converged
objective is retained. Reflection symmetry supplies the negative
half of each coordinate grid.

At fixed $\lambda$, let $\widehat{\mathcal F}(q)$ denote the retained minimum on
the coordinate grid $\mathcal Q$. The corresponding numerical
free energy profile, shifted to zero minimum, is
$
 \widehat F(q)=\widehat{\mathcal F}(q)
 -\min_{q'\in\mathcal Q}\widehat{\mathcal F}(q')$.
For the joint section, the minimizing interval is estimated from the
connected grid points whose objective lies within $10^{-6}$ of the
grid minimum, with adjacent exterior points bracketing its endpoints.

Finally, the analytical-model parameter can be related to the flow
strength used for the kinetic model. For beadwise mobility $1/\zeta$
and physical bond length $b$, matching the quadratic planar-flow
potential in units of $k_{\mathrm B}T$ to Eq.~\eqref{eq:chain-potential}
gives
\begin{equation}
\begin{aligned}
 \lambda=\frac{\zeta\dot\epsilon b^2n^2}{2k_{\mathrm B}T}
 =\frac{\zeta b^2n^2}{2k_{\mathrm B}T\tau}\mathrm{Wi},
 \qquad \mathrm{Wi}=\dot\epsilon\tau.
\end{aligned}
 \label{eq:supp-beta-wi-convention}
\end{equation}
For the $300$-bead reference chain,
$\widetilde{\dot\epsilon}=\dot\epsilon\tau_d$ and
$\tau/\tau_d=4321$, so we have $
\frac{\lambda}{n^2}=\frac{\widetilde{\dot\epsilon}}{2}$,
 $\mathrm{Wi}=\frac{8640}{n^2}\lambda$.
Using the approximation $n\simeq300$ for this reference chain,
we express the analytical results in Fig.~\ref{fig:2} in terms of $\mathrm{Wi}=\frac{8640}{300^2}\lambda=0.096\lambda$.

\section{Single-molecule experiments and image-response analysis}
\label{supp:experiments}

\subsection{Experimental protocol and image observables}
\label{supp:exp-platform}
We acquired single-molecule DNA fluorescence trajectories using the automated electrokinetic stretcher of Ref.~\cite{soh2023automated}, which controls field strength via the applied voltages and keeps the molecule near the stagnation point via active feedback control (see Ref.~\cite{soh2023automated} for experimental details).

Image coordinates and trajectory centering follow Eq.~\eqref{eq:em-exp-coordinate}. We threshold pixels at $I_p\ge20$, discard frames with fewer than $40$ foreground pixels, and remove connected components smaller than $10\%$ of the largest component to extract the projected span $L=x_{\max}-x_{\min}$. Since $x_{\min}$ and $x_{\max}$ can lie on folded interior segments rather than the true chain ends, these image observables serve as surrogates, instead of exact matches to the microscopic coordinates.

We define the reference span $L_{\max}$ as the $99$th percentile of $L$ under matched optics. We retain frames whose mask stays clear of the image boundary and changes in $L$ by less than $0.25L_{\max}$ between consecutive frames. Reference frames are further restricted to $0.30\le L(t)/L_{\max}<0.55$ prior to the first crossing of $L>0.72L_{\max}$, isolating a common stage of initial stretching (rather than imposing $Z_2=0$).

\subsection{Lag selection and response estimation}
\label{supp:exp-response}
Both future frames in each finite-lag increment (End Matter) must also pass these quality checks. At each voltage, we select the lag pair $(\delta_1,\delta_2)$ such that the mean relative span growth
\begin{equation}
 r_L(\delta_1,\delta_2)=
 \left\langle
 \frac{L(t+\delta_2)-L(t+\delta_1)}{L(t)}
 \right\rangle .
\label{eq:supp-exp-growth}
\end{equation} 
is closest to $5\%$. The table below lists the resulting lag pairs, relative growth $r_L$, and number of contributing trajectories $N$:
\begin{center} 
\begin{tabular}{lllll}
\hline
Voltage (V)~~~~~~ & $N$~~~~~~  & $\delta_1$ (frames) ~~~~~~  & $\delta_2$ (frames) ~~~~~~  & $r_L$~~~~~~ \\
\colrule
10 & 44 & 20 & 59 & 0.006\\
20 & 201 & 2 & 5 & 0.047\\
40 & 335 & 1 & 2 & 0.045\\
50 & 444 & 1 & 2 & 0.058\\
\hline
\end{tabular}
\end{center} 
The conditioning coordinate $z(t)$, increment $\Delta\widetilde Z_3(t)$, raw response $f$, and odd response $g$ follow Eq.~\eqref{eq:em-exp-response}, with all valid pairs weighted equally. By default, we use eight equal-width bins over $[-2.6\sigma_z,2.6\sigma_z]$ (where $\sigma_z$ is the standard deviation of $z$), requiring at least $800$ valid pairs overall and $150$ per retained bin. Figure~\ref{fig:supp-exp-summary}(a,b) compares $f$ and $g$ at $20$ and $40$~V, showing that the central positive slope and outer downturn at $40$~V are already evident in the raw response $f$ prior to odd symmetrization.

Odd-response estimates require both symmetric bins in a pair to be valid. We compute error bars from the $16$th--$84$th percentiles across $2000$ bootstrap resamples of the $N$ trajectories (with replacement and fixed bin edges, omitting bins with fewer than $150$ pairs). We estimate the central slope $s_0$ by a linear fit to the odd response over bins with $|z|\le\sigma_z$ (requiring at least three valid bins; four innermost bins by default), which provides a local alternative to the full-range cubic coefficient $c_1$ in the End Matter.

To test robustness across analysis parameters [Fig.~\ref{fig:supp-exp-summary}(c)], we vary the span window ($[0.28,0.52)$, $[0.30,0.55)$, $[0.32,0.58)$), target growth $r_L$ ($0.040$, $0.051$, $0.060$), bin count, and binning range ($\pm2.3\sigma_z$, $\pm2.6\sigma_z$, $\pm2.9\sigma_z$). Among the $108$ combinations, all $72$ admissible settings yield $s_0<0$ at $20$~V and $s_0>0$ at $40$~V, confirming a robust sign reversal of the central slope despite sampling uncertainty at $40$~V; the outer zero crossing is resolved in $54$ of the $72$ settings. 

\begin{figure}[tbp]
\centering
\includegraphics[width=0.7\linewidth]{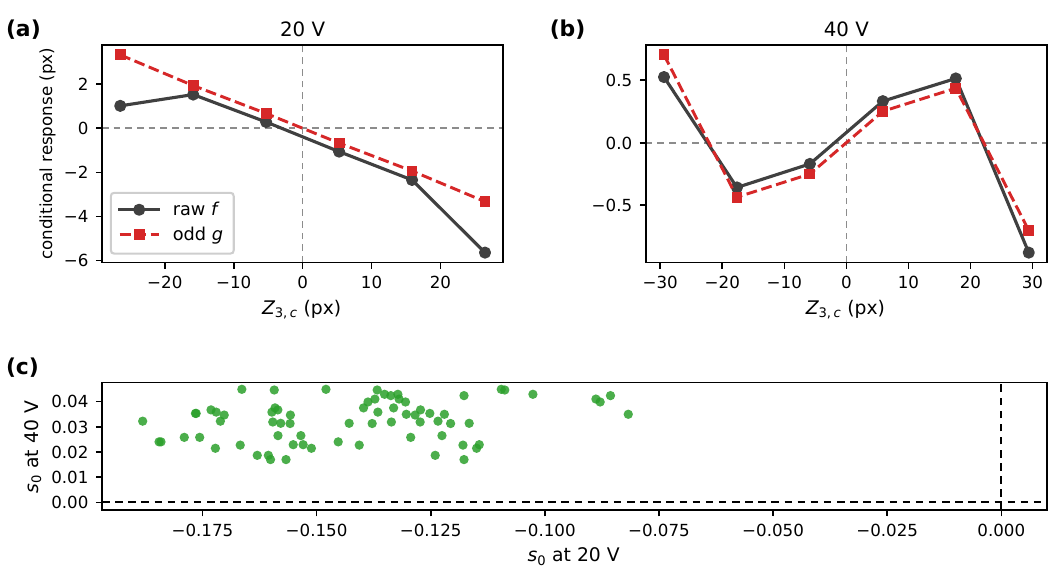}
\caption{Experimental response checks and sensitivity to analysis
choices. (a,b) Raw conditional response $f$ (gray circles) and odd
component $g$ (red squares)
at $20$ and $40$~V. (c) Central
slopes $s_0$ at $20$ and $40$~V for the $72$ admissible analysis
settings.}
\label{fig:supp-exp-summary}
\end{figure}

\end{document}